\documentclass[preprint,12pt,authoryear]{elsarticle}

\makeatletter
\def\ps@pprintTitle{%
     \let\@oddhead\@empty
     \let\@evenhead\@empty
     \def\@oddfoot{\footnotesize\itshape
       Published in \ifx\@journal\@empty Elsevier
       \else\@journal\fi\hfill\today}%
     \let\@evenfoot\@oddfoot}
\makeatother

\usepackage{amsmath}
\usepackage{amssymb}
\usepackage{amsthm}
\usepackage{csquotes}
\usepackage{bbm}
\usepackage[dvipsnames]{xcolor}
\usepackage{enumitem}
\usepackage{booktabs}
\usepackage{graphicx}
\usepackage[caption=false,font=footnotesize]{subfig}

\usepackage{natbib}
\usepackage{url} 

\newcommand{\bs}[1]{\boldsymbol{#1}}
\def\Rbb{\mathbb{R}}
\def\Ebb{\mathbb{E}}
\def\Pbb{\mathbb{P}}

\def\Hr{\mathcal{H}}

\def\hats{\hat{s}}
\def\xbs{\bs{x}}
\def\Xbs{\bs{X}}

\newcommand{\pval}[1]{$p$-value}
\newcommand{\pvals}[1]{$p$-values}

\newcommand{\Idisc}{I_{\rm discrete}}

\newcommand{\rev}[1]{\textcolor{black}{#1}}

\newtheorem{theorem}{Theorem}

\newtheorem{lemma}{Lemma}
\newtheorem{corollary}{Corollary}
\newtheorem{definition}{Definition}
\newtheorem{remark}{Remark}

\newcommand{\Dtest}{\mathcal{D}^{\rm test}}
\newcommand{\Dnull}{\mathcal{D}^{\rm null}}
\definecolor{color1}{HTML}{0011af}
\definecolor{color2}{HTML}{8819a0}
\definecolor{color3}{HTML}{bf418d}
\definecolor{color4}{HTML}{e37076}
\definecolor{color5}{HTML}{f9a256}
\definecolor{color6}{HTML}{FF0000}
\definecolor{color7}{HTML}{009F6B}
\definecolor{color8}{HTML}{00CC99}

\usepackage{glossaries}

\newacronym{fwer}{FWER}{family-wise error rate}
\newacronym{fdr}{FDR}{false discovery rate}
\newacronym{tdr}{TDR}{true discovery rate}
\newacronym{bh}{BH}{Benjamini-Hochberg}
\newacronym{iid}{iid}{independent and identically distributed}
\newacronym{erl}{ERL}{extreme rank length}
\newacronym{prds}{PRDS}{positive regression dependent on a subset}
\newacronym{get}{GET}{global envelope test}
\newacronym{mmctest}{MMCTest}{multiple Monte Carlo test}
\newacronym{cmmctest}{CMMCTest}{conformal multiple Monte Carlo test}

\usepackage{hyperref}
\hypersetup{colorlinks=true}

\journal{Spatial Statistics}

\begin{document}


\begin{frontmatter}



\title{Conformal novelty detection for replicate point patterns with FDR or FWER control} 

\author[label1]{Christophe~A.~N.~Biscio} 
\author[label2]{Adrien~Mazoyer} 
\author[label1]{Martin~V.~Vejling\corref{cor1}} 

\cortext[cor1]{Main and corresponding author.}

\affiliation[label1]{organization={Department of Mathematical Sciences, Aalborg University},
            addressline={TMV 23},
            city={Aalborg},
            postcode={9000}, 
            country={Denmark}}

\affiliation[label2]{organization={Institut de Mathematiques de Toulouse, Universite de Toulouse},
            addressline={UPS}, 
            city={Toulouse},
            postcode={F-31062}, 
            country={France}}

\begin{abstract}
Monte Carlo tests are widely used for computing valid p-values without requiring known distributions of test statistics. When performing multiple Monte Carlo tests, it is essential to maintain control of the type I error. Some techniques for multiplicity control pose requirements on the joint distribution of the p-values, for instance independence, which can be computationally intensive to achieve%
{, as it requires simulating disjoint null samples for each test. We refer to this as naïve multiple Monte Carlo testing}. We highlight in this work that multiple Monte Carlo testing is an instance of conformal novelty detection. Leveraging this insight enables a more efficient multiple Monte Carlo testing procedure, avoiding excessive simulations {by using a single null sample for all the tests,} while still ensuring exact control over the false discovery rate or the family-wise error rate. We call this approach conformal multiple Monte Carlo testing. The performance is investigated in the context of global envelope tests for point pattern data through a simulation study and an application to a sweat gland data set. Results reveal that with a fixed {simulation budget},
our proposed method yields substantial improvements in power of the testing procedure as compared to the naïve multiple Monte Carlo testing procedure.
\end{abstract}




\end{frontmatter}



\section{Introduction} \label{sec:introduction}

We study a novelty detection problem  where we have access to a null sample $\Dnull = \{\Xbs_i\}_{i=1}^{n}$ of independent data points generated under a distribution $P_0$ on a space $\mathcal{X}$, and a test sample $\Dtest =\{\Xbs_{n+j}\}_{j=1}^{m}$ where $\Xbs_{n+j}\sim P_j$ for {distributions} $P_j$ on the space $\mathcal{X}$, for $j=1,\dots,m$. We call each $\Xbs_{n+j}$ a test point, and consider specifically the case where $\mathcal{X}$ is the space of locally finite point configurations on a bounded domain. There the aim is to test which observations in $\Dtest$ are novelties, i.e. $P_j \neq P_0$, for $j=1,\dots,m$, while controlling for false positives that can be extremely harmful in practice.

This setting is common in spatial statistics where we encounter replicated point pattern data \citep{Diggle2000:Comparison}. Examples are locations of pyramidal neurons \citep{Diggle1991:Analysis}, cells \citep{Baddeley1993:Analysis}, epidermal nerve fibers \citep{Konstantinou2023:Statistical}, sweat glands \citep{Kuronen2021:Point}, and more \citep{Baddeley2015:Spatial}. For instance the sweat gland data studied by \cite{Kuronen2021:Point} consisting in five point patterns for each of the following three groups of subjects: diagnosed with neuropathy, suspected to have neuropathy, control (assumed to not have neuropathy). In this case, we would be interested in deriving a multiple testing procedure for the subjects suspected to have neuropathy, based on the data from the control subjects and the diagnosed subjects.

Our setting can be compared to the ones of~\cite{Gandy2014:MMCTest,Gandy2017:Quick}, and \cite{Hahn2020:Optimal}, although we differ in part due to our focus on spatial data. There, for each $j$, testing for $P_j \neq P_0$ is done by Monte Carlo tests which are based on a null sample constructed from independent simulations under $P_0$, see \cite{Baddeley2015:Spatial} and \cite{Myllymaki2017:Global}. This results in independent \pvals{}, and procedures such as the Hochberg procedure by \cite{Hochberg1988:Sharper} or the Benjamini-Hochberg procedure by \cite{Benjamini1995:FDR} can be used to ensure bounds on the \gls{fwer} or the \gls{fdr}, respectively. However, this comes with a significant computational strain, in particular if the number of test points $m$ is large and simulation under $P_0$ difficult, as for each $j=1,\ldots,m$ we need to have access to a new null sample. {We refer to this as the \textit{naïve multiple Monte Carlo testing} approach.}

Adopting another approach, \cite{Bates2023:Testing} and \cite{Marandon2024:Adaptive} re-use simulations under $P_0$ for different test points thereby resulting in dependent \pvals{} which in general cannot control the \gls{fwer} or the \gls{fdr}. In particular, \cite{Bates2023:Testing} introduced conformal \pvals{}, and proved that they are \gls{prds}, controlling the \gls{fdr} under the \gls{bh} procedure \citep{Benjamini2001:Dependency}.
Further, they proved that the \gls{bh} procedure with Storey's correction \citep{Storey2002:Direct, Storey2004:Strong} also controls the \gls{fdr}.

Inspired by this approach, we propose a \textit{conformal multiple Monte Carlo testing} procedure {that uses the same null sample for multiple Monte Carlo tests, instead of the naïve approach that would use disjoint null samples for each test.}
It is based on conformal \pvals{} computed for each test point in $\Dtest$ which are then combined by Storey's \gls{bh} procedure (resp.~the Hochberg procedure). Using recent advances of conformal novelty detection in \cite{Bates2023:Testing} and \cite{Marandon2024:Adaptive}, we prove that the methods control the \gls{fdr} (resp.~the \gls{fwer}) and present a rigorous procedure for testing goodness-of-fit for replicate point patterns. In doing so we extend the \gls{get} in \cite{Myllymaki2017:Global} to replicate point patterns, maintaining a graphical interpretation, while highlighting connections between conformal methodology and Barnard's Monte Carlo test \citep{Barnard1963:Discussion}.
On simulated data and a dataset from biology on the position of sweat glands among patients with neuropathy, we illustrate that {with a given simulation budget} our method significantly improves the power of testing procedures for replicate point patterns, compared to a naïve usage of \cite{Myllymaki2017:Global}.

The rest of the paper is organized as follows. Section~\ref{sec:background} presents preliminaries on multiple testing as well as the recent developments surrounding conformal \pvals{}, and we present in Section~\ref{sec:controlling_FWER} our contribution to \gls{fwer} control with conformal \pvals{}. Then, Section~\ref{sec:spatialdata} outlines how conformal \pvals{} can be used for multiple testing in the context of spatial point processes, herein relating conformal novelty detection to the modern \gls{get} in the domain of spatial statistics. Section~\ref{sec:simulation} presents simulation studies illustrating the power of the proposed methodology, and in Section~\ref{sec:real_data}, we apply it to the sweat gland data from \cite{Kuronen2021:Point}. Finally, Section~\ref{sec:conclusion} concludes the paper and presents some future research directions.

\section{Background}\label{sec:background}

\subsection{Multiple hypothesis testing}\label{sec:multiplehyptesting}
We consider the multiple hypothesis testing scenario in which $m$ test points are observed and a null hypothesis is tested on each of them. In other words, $m$ null hypotheses $H^{1}_0,\dots, H^{m}_0$ are tested at level $q^*\in(0, 1)$, and we denote by $p_1, \dots, p_m$ the corresponding \pvals{}.

We denote by $S$ the number of rejected non-true null hypotheses, and by $V$ the number of rejected true null hypotheses (type I error). Moreover, $\mathcal{R}=V+S$ denotes the total number of rejected hypotheses, while $m_0$ denotes the number of true null hypotheses.

A classical error rate studied for multiple comparisons is the \gls{fdr} criterion, proposed by \cite{Benjamini1995:FDR}, which is defined as $\mathrm{FDR}=\Ebb[\mathrm{FDP}]$, where $\mathrm{FDP}~=~V/\max(1, \mathcal{R})$. The power of a test, also called the \gls{tdr}, is defined as $\Ebb[{\rm TDP}]$, where the true discovery proportion is ${\rm TDP} = S/\max(1, m-m_0)$. The most known method to control the \gls{fdr} is the \gls{bh} procedure presented below.

\begin{definition}[\gls{bh} procedure \citep{Benjamini1995:FDR}]\label{def:BH}
    Let $p_{(1)}\leq p_{(2)} \leq \cdots \leq p_{(m)}$ be the ordered \pvals{} and denote by $H^{(j)}_0$ the null hypothesis associated to $p_{(j)}$. The \gls{bh} procedure with level $q^*\in(0,1)$ is to reject $H^{(1)}_0,\dots, H^{(k)}_0$ where 
    \begin{equation*}
        k = \max_{j\in\{1,\dots,m\}}\Big\{p_{(j)}\leq \frac{j}{m}q^*\Big\}.
    \end{equation*}
\end{definition}

When the \pvals{} are independent, the \gls{bh} procedure controls the \gls{fdr} at level $\pi_0 q^*$ where $\pi_0= m_0/m$ is the proportion of true nulls, see \cite{Benjamini1995:FDR}. We present below a more general condition under which the \gls{bh} procedure also controls the \gls{fdr} at level $\pi_0 q^*$, namely \gls{prds}.
\begin{definition}[PRDS \citep{Benjamini2001:Dependency}]
    A random vector $\Xbs=(X_i)_{i=1\dots m}$ is PRDS on $\mathcal{H}_0\subset\{1,\dots,m\}$ if for any increasing set {$\mathcal{D}\subset\Rbb^m$  (see below)}, and for each $i \in \mathcal{H}_0$, the function $\Pbb(\Xbs \in \mathcal{D} | X_i = x)$ is nondecreasing in $x$.
\end{definition}
{A set $\mathcal{D}\subset\Rbb^m$ is increasing if $(x_1,\dots,x_m) \in \mathcal{D}$ and $x_i\leq y_i$ for all $i=1\dots m$ implies $(y_1,\dots,y_m) \in \mathcal{D}$.}
%
\begin{theorem}[Theorem 1.2 \citep{Benjamini2001:Dependency}]\label{theorem:BH_PRDS}
    If the joint distribution of \pvals{} is PRDS on $\mathcal{H}_0 = \{j\in\{1,\dots,m\}\;\mathrm{s.t.}\;H_0^{j}\;\mathrm{is}\;\mathrm{true}\}$, then the \gls{bh} procedure applied with level $q^*\in(0,1)$ controls the \gls{fdr} at level less than or equal to $\pi_0 q^*$:
    \begin{equation*}
         {\rm FDR} = \Ebb\bigg(\frac{V}{\max(1, \mathcal{R})}\bigg) \leq \pi_0 q^* \leq q^*.
    \end{equation*}
\end{theorem}
To avoid being overly conservative, we would need to know $\pi_0$ and set $q^* = \alpha/\pi_0$ to have global significance level $\alpha \in (0, 1)$.
However, in practice $\pi_0$ is not known, and several estimation procedures have been proposed (see e.g. \cite{Storey2002:Direct}, \cite{Storey2004:Strong} and \cite{Benjamini2006:Adaptive}), without one being constantly better. We present below the Storey estimator as it adapts well in our conformal inference setting (see Theorem~\ref{theorem:conformal_Storey}).
\begin{definition}[Storey's estimator \citep{Storey2004:Strong}]\label{def:storey}
    The Storey estimator of the proportion of true nulls $\pi_0$ with parameter $\lambda\in {(0,1)}$ is 
    \begin{equation*}
        \hat{\pi}_0 = \frac{1+ \sum_{i=1}^m\mathbbm{1}[p_i > \lambda]}{m(1-\lambda)}.
    \end{equation*}
\end{definition}
The idea is that the \pvals{} contain information regarding the number of true null hypotheses among the $m$ hypotheses. The hyper-parameter $\lambda$ has a classic influence on $\hat{\pi}_0$ in terms of bias-variance trade-off. A bootstrap procedure for choosing the optimal $\lambda$ is proposed by  \cite{Storey2002:Direct} in the case of independent \pvals{}, but in general, the \gls{bh} procedure used with an estimator of $\pi_0$ does not control the \gls{fdr} at the prescribed level in case of PRDS \pvals{}. We present in Section~\ref{sec:conformal_outlier_detection} how conformal \pvals{}, used in combination with Storey's estimator from Definition~\ref{def:storey} allows to control the \gls{fdr}.

A different error criterion for multiple testing which is more strict than \gls{fdr} is the \gls{fwer}, defined as $\Pbb(V \geq 1)$. The traditional Bonferroni procedure for controlling the \gls{fwer} is to reject any hypothesis $H^{j}_0$, $j=1,\dots,m$, for which the corresponding \pval{} $p_j \leq q^*/m$ for significance level $q^*\in(0,1)$ \citep{Bonferroni1936:Teoria}. Equipped with a liberal estimator of $\hat{m}_0$, for instance Storey's estimator, a sharper Bonferroni procedure is to reject when $p_j \leq q^*/\hat{m}_0$. Multiple testing procedures for controlling the \gls{fwer}, which are uniformly more powerful than the traditional Bonferroni procedure, have been proposed under certain dependency requirements, for instance the Hochberg procedure presented below \citep{Simes1986:Improved,Hochberg1988:Sharper}.
\begin{definition}[Hochberg procedure \citep{Hochberg1988:Sharper}]\label{def:Hochberg}   
    Let $p_{(1)}\leq p_{(2)} \leq \cdots \leq p_{(m)}$ be the ordered \pvals{} and denote by $H^{(j)}_0$ the null hypothesis associated to $p_{(j)}$. The Hochberg procedure with level $q^*\in(0,1)$ is to reject $H^{(1)}_0,\dots, H^{(k)}_0$ where 
    \begin{equation*}
        k = \max_{j\in\{1,\dots,m\}}\Big\{p_{(j)}\leq \frac{q^*}{m-j+1}\Big\}.
    \end{equation*}
\end{definition}

When the \pvals{} are independent, or when their joint distribution has a specific type of positive dependence called multivariate totally positive of order two (${\rm MTP}_2$, see Definition~\ref{def:mtp2}){, the Hochberg procedure controls the \gls{fwer} at level $q^*\in(0,1)$} \citep{Sarkar1997:Simes,Sarkar1998:Probability}.

\subsection{Conformal p-values}\label{sec:conformal_outlier_detection}
A conformal inference approach to novelty detection has been studied in recent works under different names, see \cite{Bates2023:Testing,Mary2022:Semi,Marandon2024:Adaptive}.
The central assumption there is that the random variables in $\Dnull$ are exchangeable as defined below. 
\begin{definition}[Exchangeability]\label{def:exchangeability}
    The random variables $\Xbs_1, \dots, \Xbs_n$ are exchangeable if their joint distribution is invariant by permutations.
\end{definition}

Let $\hats$ be a data driven score function mapping from $\mathcal{X}\times\mathcal{X}^{n+m}$ to a {real set}. The space $\mathcal{X}$ is traditionally $\mathbb{R}^d$ or $\mathbb{R}^{d\times d'}$, see for instance \cite{Marandon2024:Adaptive}, but in this work $\mathcal{X}$ will be the space of locally finite point configurations on a bounded domain. As usual in the conformal framework, this score is data driven and we will be particularly interested, for each $\Xbs_i\in\Dnull\cup\Dtest$, in the evaluation of $\hats_i = \hats(\Xbs_i, (\Xbs_{1}, \dots, \Xbs_{n+m}))$.
We denote by $\Hr_0 = \{j\in\{1,\dots,m\}\;\mathrm{s.t.}\;H_0^{j}\;\mathrm{is}\;\mathrm{true}\}$ the set of nulls in the test sample $\Dtest$, $\Hr_1 = \{1,\dots,m\}\setminus\Hr_0$, and make the following assumptions on the score and the data:
\begin{enumerate}[label=({$\mathcal{A}$\arabic*})]
    \item \label{A:invariance} for any permutation $\pi$ of $\{1,\dots,n+m\}$, $\hats$ satisfies for any $(\Xbs_{1}, \dots, \Xbs_{n+m})\in\mathcal{X}^{n+m}$ the invariance property $\hats(\cdot, (\Xbs_{\pi(1)}, \dots, \Xbs_{\pi(n+m)})) = \hats(\cdot, (\Xbs_{1}, \dots, \Xbs_{n+m}))\,;$
    \item $(\Xbs_1,\dots,\Xbs_n,(\Xbs_{n+i})_{i\in\Hr_0})$ are exchangeable conditional on $(\Xbs_{n+i})_{i\in\Hr_1}$;\label{A:exchangeability}
    \item there are almost surely no ties among $(\hats_1, \ldots, \hats_{n+m})$, i.e., $\mathbb{P}(\hats_i < \hats_j \lor \hats_j < \hats_i) = 1$ for all $i\neq j$.\label{A:notie}
\end{enumerate}
For $j=1,\ldots, m$, we define the conformal \pval{} associated to the test point $\Xbs_{n+j}$ by
\begin{equation}\label{eq:marginal_conformal_p_values_continuous}
    \hat{p}_j = \frac{1}{n+1} \Big(1 + \sum_{i={1}}^{n} \mathbbm{1}[\hats_i \leq \hats_{n+j}]\Big).
\end{equation}
The conformal framework can be extended to conformal scores $\hats$ taking values in a non-scalar space equipped with an ordering $\preceq$. Then, the p-values can be defined as \eqref{eq:marginal_conformal_p_values_continuous} with $\leq$ replaced by $\preceq$ \citep{Myllymaki2017:Global}.

Notice that \cite{Marandon2024:Adaptive} exposes a more general case where the score function is learned: the null sample $\Dnull$ is then split in two parts $(\Xbs_1,\dots,\Xbs_\ell)$ and $(\Xbs_{\ell+1},\dots,\Xbs_{n+m})$, where the former is used only for learning $\hats$, and the invariance assumption \ref{A:invariance} becomes
\begin{equation*}
        \hats(\cdot, (\Xbs_{1},\dots,\Xbs_{\ell}), (\Xbs_{\pi(\ell+1)}, \dots, \Xbs_{\pi(n+m)})) = \hats(\cdot, (\Xbs_{1}, \dots, \Xbs_{\ell}), (\Xbs_{\ell+1},\dots,\Xbs_{n+m}))\,,
\end{equation*} 
and the conformal \pvals{} are given by $\hat{p}_j = \frac{1}{n-\ell+1} \big(1+\sum_{i=\ell+1}^{n} \mathbbm{1}[\hats_i \leq \hats_{n+j}]\big)$.
In Section~\ref{sec:spatialdata}, we view conformal novelty detection in the context of \gls{get} for replicate point patterns, and here we do not need $\ell > 0$, so for the sake of readability we consider hereinafter $\ell=0$.
A less general framework was initially considered in \cite{Bates2023:Testing} where the conformal scores were assumed to depend only on $\Xbs_{1},\dots,\Xbs_{\ell}$. For this case they showed that the conformal \pvals{} are PRDS, which implies that applying \gls{bh} procedure (Definition~\ref{def:BH}) controls the \gls{fdr}.
For the general case, \cite{Marandon2024:Adaptive} have found lower and upper bounds for \gls{fdr}, as stated in the theorem below.
\begin{theorem}[\gls{fdr} control \citep{Marandon2024:Adaptive}]\label{thm:FDR_control}
    Let assumptions~\ref{A:invariance}-\ref{A:notie} be true for a conformal score $\hats$. The \gls{bh} procedure (see Definition~\ref{def:BH}) at level $q^*$ with the conformal \pvals{} satisfies
    \begin{equation*}
        m_0 \frac{\lfloor q^* (n{+}1) / m \rfloor}{(n{+}1)} \leq \Ebb\left(\frac{V}{\max(1, \mathcal{R})}\right) \leq \pi_0 q^*,
    \end{equation*}
    and
    \begin{equation*}
        \Ebb\left(\frac{V}{\max(1, \mathcal{R})}\right) = \pi_0 q^*,
    \end{equation*}
    when $q^* (n+1)/m$ is an integer.
\end{theorem}

Moreover, \cite{Marandon2024:Adaptive} showed that the \gls{fdr} is controlled in the case of Storey's correction, which we state in Theorem \ref{theorem:conformal_Storey} below.
\begin{theorem}[Storey's \gls{bh} with conformal \pvals{} \citep{Marandon2024:Adaptive}]\label{theorem:conformal_Storey}
    Let $\hat{p}_j$, for $j=1,\ldots,m$, be conformal \pvals{} given by~\eqref{eq:marginal_conformal_p_values_continuous} with a conformal score verifying~\ref{A:invariance}-\ref{A:notie}. Let further $\hat{\pi}_0$ be Storey's estimator as in Definition~\ref{def:storey} {with $\lambda=K/(n+1)$ for $K \in \{1, \dots, n\}$}. Then, the \gls{bh} procedure, as in Definition~\ref{def:BH}, applied at level $q^* = \alpha/\hat{\pi}_0$ to $\hat{p}_1,\ldots,\hat{p}_m$ controls the \gls{fdr} at the prescribed significance level $\alpha \in (0, 1)$.
\end{theorem}
We note that other adaptive \gls{bh} methods can be used with conformal p-values, for instance the method of \cite{Benjamini2006:Adaptive}, as proven in \cite{Marandon2024:Adaptive}. We restrict our attention in this work to Storey's estimator as this was initially proven to work with conformal p-values in \cite{Bates2023:Testing}.

\section{Controlling the FWER with conformal p-values}\label{sec:controlling_FWER}
While \gls{fdr} is a more common and popular multiple testing criterion, some situations still warrants the use of the \gls{fwer}, i.e., in cases where a practitioner is really interested in controlling the probability of making one or more false rejections.

\cite{Angelopoulos2024:Theoretical} highlighted through Proposition 10.2 that 
the \rev{$\check{\rm S}$id\'{a}k} procedure, i.e.,
rejecting hypothesis $j$ when ${\hat{p}}_j \leq 1 - (1-\alpha)^{1/m}$ for significance level $\alpha\in(0,1)$, controls the \gls{fwer} at the nominal level, and that asymptotically as the size of the null sample $n$ increases, the \gls{fwer} converges exactly to $\alpha$. Yet, it would be reasonable to expect that one could construct a more powerful procedure in the case of finite $n$ which does not use a constant threshold for all \pvals{}, but instead uses a sequence of thresholds on the ordered \pvals{}. 

It was shown by \cite{Gazin2024:Transductive} that the joint probability mass function of the conformal \pvals{} is
\begin{equation}\label{eq:joint_mass}
    \mathbb{P}(\hat{\bs{p}} = \bs{j}/(n+1)) = \frac{n!}{(n+m)!} \prod_{k=1}^{n+1} M_k(\bs{j})!,
\end{equation}
for $\bs{j} \in \{1,\dots,n+1\}^m$ and where $M_k(\bs{j}) = |\{i \in \{1,\dots,m\} : j_i = k\}|$.

We recall that if the joint distribution of \pvals{} is ${\rm MTP}_2$, then the Hochberg procedure, see {Definition}~\ref{def:Hochberg}, controls the \gls{fwer}. Since conformal \pvals{} are positively correlated (Lemma 1 of \cite{Bates2023:Testing}) and ${\rm MTP}_2$ is a quite large class of distributions \citep{Sarkar1997:Simes}, one could expect that conformal \pvals{} are ${\rm MTP}_2$. In the supplementary material, we show that this is not the case through a counter example based on the explicit form~\eqref{eq:joint_mass} (see Theorem~\ref{thm:not_mtp2}).

\rev{Through a discussion with Pierre Neuvial, it has been revealed to the authors, that the \gls{fdr} control of the \gls{bh} procedure with the conformal p-values implies Simes' inequality for conformal p-values \cite{Sarkar2008:Simes}. In turn, this implies that any closed testing procedure using a Simes local test at level $\alpha$ controls the \gls{fwer} at level $\alpha$. As a consequence, the Hochberg procedure and also the more powerful Hommel procedure (see \cite{Meijer2018:Hommel}) with conformal p-values both control the \gls{fwer} at the nominal level.}
\begin{corollary}\label{cor:conformal_hochberg}
    \rev{Let assumptions~\ref{A:invariance}-\ref{A:notie} be true for a conformal score $\hats$.
    Then, the Hochberg procedure (See Definition~\ref{def:Hochberg}) at level $\alpha$ with the conformal \pvals{} satisfies
    \begin{equation*}
        {\rm FWER} = \mathbb{P}(V \geq 1) \leq \alpha.
    \end{equation*}}
\end{corollary}
\begin{proof}
    \rev{Simes' inequality is implied by the \gls{fdr} control of the \gls{bh} procedure with conformal p-values. This in turn implies, than any closed testing procedure based on a Simes local test, herein the Hochberg procedure, controls the \gls{fwer}.}
\end{proof}
\begin{remark}
    In the supplementary material, we show that the joint distribution of the ordered \pvals{} is uniform on an integer order set, see Lemma~\ref{lem:conformal_joint_order}. Moreover, Lemma~\ref{lem:grid_poly} in the supplementary material allows for a simple recursive expression for the ${\rm FWER}$, representing a computationally efficient way to numerically evaluate the ${\rm FWER}$.
\end{remark}
Using Lemma~\ref{lem:grid_poly} in the supplementary material, we can numerically compute the theoretical ${\rm FWER}$ for any sequence of thresholds applied on the ordered \pvals{}. We show then in Figure~\ref{fig:polyFWER} the ${\rm FWER}$, for a scenario with $m=10$ and $m_0=5$ where the $m_0$ largest p-values are the true nulls, for significance levels $\alpha=0.05,0.1$, with the Hochberg procedure and the \rev{$\check{\rm S}$id\'{a}k} procedure rejecting at level $1-(1-\alpha)^{1/\hat{m}_0}$ with $\hat{m_0}=m_0$ and $\hat{m_0}=m_0+1$, against the null sample size, $n$. Here an estimator of $m_0$ is used to sharpen the {\rev{$\check{\rm S}$id\'{a}k}} procedure, and as for the traditional Bonferroni procedure this maintains \gls{fwer} control assuming a liberal estimator of $m_0$ is used, such as Storey's estimator~\cite{Storey2004:Strong}.
Figure \ref{fig:polyFWER} shows the conservativeness of the Hochberg procedure for conformal p-values. Moreover, it highlights that the Hochberg procedure is not always better than the \rev{$\check{\rm S}$id\'{a}k} procedure, and vice versa. The over-conservativeness of the \rev{$\check{\rm S}$id\'{a}k} procedure is highly dependent on the quality of the estimator $\hat{m}_0$, while for Hochberg it depends on the order position of the true nulls. We conjecture that a procedure which is uniformly more powerful than both the \rev{$\check{\rm S}$id\'{a}k} procedure and the Hochberg procedure exists, and leave this as a problem for future investigation.
\begin{figure}
    \centering
    \includegraphics[page=10, width=0.55\linewidth]{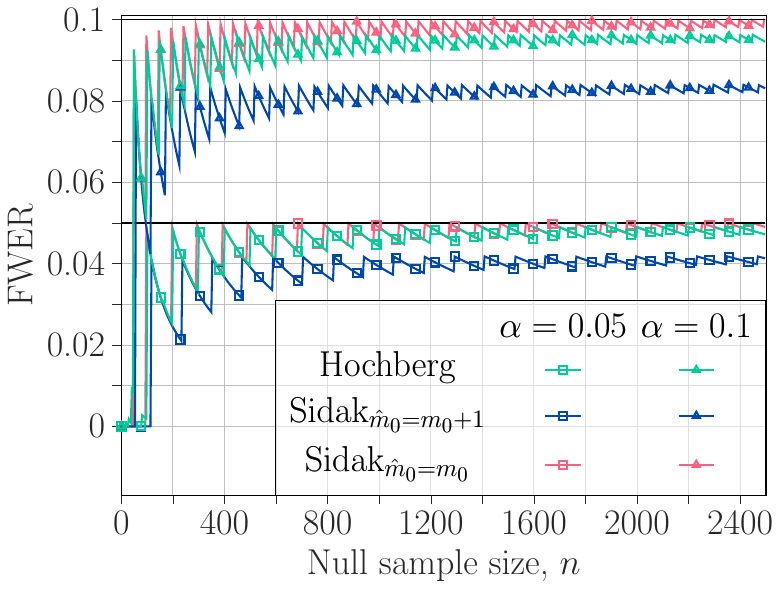}
    \caption{${\rm FWER}$ against the null sample size, $n$.}
    \label{fig:polyFWER}
\end{figure}

\section{Application to spatial point patterns}\label{sec:spatialdata}

\subsection{Point patterns and global envelope tests} \label{sec:pp_envtest}

Spatial points processes are used in various fields such as ecology, social sciences or biology to study the spatial distribution and interaction between objects/events of interest, trees, or cells, see \cite{Baddeley2015:Spatial}.

Mathematically, a point process is a random configuration of points such that in any bounded set there is almost surely a finite number of points. Alternatively, one can see a point process $\Xbs$ as a random set $\lbrace Z_1,\ldots,Z_N \rbrace$ where $N$ is an integer valued random variable, possibly infinite. For a rigorous mathematical presentation on point processes, we refer the reader to \cite{Last2017:LecturePointProcess}.

Due to their complex nature, relevant statistics for point processes that could be used for novelty detection are not scalar but functions, which often depends on the inter point distance, e.g., Besag's $L$-function or Ripley's $K$-function (see e.g. \cite{MollerWaagepeterson04}). The problem with constructing tests based on such functional statistics is that one would need to conduct a global test for each inter point distance in a set of discretization points which requires either a multiple testing procedure or to reduce the functional statistic to a scalar. Examples of such methods are the maximum absolute deviation test and the Diggle-Cressie-Loosmore-Ford test which are respectively adaptations of Kolmogorov-Smirnoff tests and Cramér-Von Mises tests to spatial statistics \citep{Baddeley2015:Spatial}. However, this drastically reduce the power of the tests we could define. A modern approach to tackle goodness-of-fit for point patterns is based on \gls{get}, a functional version of Barnard's Monte Carlo test \citep{Myllymaki2017:Global}. A comprehensive review on goodness-of-fit tests for spatial point processes is given by \cite{Fend2025:Goodness}.

Barnard's Monte Carlo \citep{Barnard1963:Discussion} test can be summed up as following. Consider a statistic $T_{n+1}$ associated to the test point on which we want to test the goodness-of-fit of a null model $P_0$, and $n$ \gls{iid} statistics $T_1,\dots,\,T_n$, computed from $n$ \gls{iid} points simulated under $P_0$. Recall that each data point in this setup is a point pattern. Under the null hypothesis, $T_1,\dots,T_{n+1}$ are \gls{iid} and $T_{n+1}$ is among the $\lfloor\alpha(n+1)\rfloor$ largest $T_i$ with probability $\alpha$.
Notice that the distribution of the $T_i$ does not need to be known. Moreover, as in the conformal framework, we can replace the \gls{iid} assumption by exchangeability of $T_1,\dots,\,T_{n+1}$.

In \gls{get}, $T_i$ are functional statistics, typically the $L$-function. Applying the above methodology requires then to define an order relation between the $T_i$.
In 
{\cite{Mrkvicka2022:New}}, the rank measure is defined as follows. Consider $\Xbs_1,\dots,\Xbs_{n}$ \gls{iid} from a null distribution $P_0$, and $\Xbs_{n+1}$ a test point for which we want to test if $\Xbs_{n+1} \sim P_0$. Let $T_i$ be a functional summary statistic for the point pattern $\Xbs_i$ defined on an interval $I \subset \Rbb$.
For $r\in I$, we denote by $R^{\uparrow}_i(r)$ (respectively $R^{\downarrow}_i(r)$) the ascending (descending) rank of $T_i(r)$ among $\mathcal{T}_{1:n+1} = \{T_1(r),\dots, T_{n+1}(r)\}$.  {Pointwise p-values are $u_i(r) = R_i^*(r)/(n+1)$} where $R_i^*(r) = \min \big\{R^{\uparrow}_i(r), R^{\downarrow}_i(r)\big\}${, and the minimum of the pointwise p-values is defined as $u^{\rm min}_i = \min_{r\in I} u_i(r)$.}
By construction, there will be ties among the {$u^{\rm min}_i$}. The ordering of the curves $T_i$ based on this rank measure is therefore only weak.
To reduce the number of ties, \cite{Myllymaki2017:Global} introduced another ordering called the \gls{erl}.
{Consider a discretization of $I$, denoted $\Idisc = \{r_1,\dots,r_M\}$, resulting in the pointwise p-values $u_{i,j} = u_i(r_j)$, $j=1,\dots,M$. Define then the vector of ordered pointwise p-values $\bs{u}_i = (u_{i,(1)}, \dots, u_{i,(M)})$ such that $u_{i,(1)} \leq u_{i,(2)} \leq \dots \leq u_{i,(M)}$. The extreme rank length ordering is defined as $\bs{u}_j \prec \bs{u}_i$ if}
\begin{equation*}
    {\exists\,l \leq M : u_{j,(k)} = u_{i,(k)}\, \quad \forall\,k < l\,,\; u_{j,(l)}<u_{i,(l)},}
\end{equation*}
${\bs{u}}_j \equiv {\bs{u}}_i$ if neither ${\bs{u}}_j \prec {\bs{u}}_i$ nor ${\bs{u}}_i \prec {\bs{u}}_j$, and denote by ${\bs{u}}_j \preceq {\bs{u}}_i$ that either ${\bs{u}}_j \prec {\bs{u}}_i$ or ${\bs{u}}_j \equiv {\bs{u}}_i$.
Ties are then still possible among the ${\bs{u}}_i$, but are much less likely, in particular if $\Idisc$ is large.

The \gls{erl} is used to define a global extreme rank envelope test for the null hypothesis $\Xbs_{n+1} \sim P_0$.
The \pval{} of the test is
\begin{equation}\label{eq:get_pvalue}
    \hat{p} = \frac{1}{n+1}\Big(1+\sum_{i=1}^{n} \mathbbm{1}[{\bs{u}}_i \preceq {\bs{u}}_{n+1}]\Big)\,,
\end{equation}
which turns out to be a conformal \pval{} with the conformity score: 
\begin{equation}\label{eq:ERL conformity score}
    \hats^{\rm erl}_i \equiv \hats^{\rm erl}(\Xbs_i, (\Xbs_1, \dots,\Xbs_{n}, \Xbs_{n+1})) = {\bs{u}}_i.
\end{equation}
\begin{remark}
    The \gls{erl} measure depends on the set of summary statistics used for the ordering, $\mathcal{T}_{1:n+1}$, such that a more rigorous notation is ${\bs{u}}_i \equiv {\bs{u}}_i^{\mathcal{T}_{1:n+1}}$. We shall require this in the following section.
\end{remark}

An important hyperparameter for global extreme rank envelope is the choice of the size $n$ of $\Dnull$. Some recommendations can be found in \cite{Myllymaki2017:Global}, \cite{Mrkvicka2017:Multiple}, and other related works. Generally speaking, an appropriate choice of $n$ depends on many factors, among others $|\Idisc|$, but a repeated recommendation has been to use $n=2500$, which is costly when repeated for each test point.

One of the strengths of the \gls{get} is the graphical interpretation \citep{Myllymaki2017:Global}: a $100(1-\alpha)~\%$ global envelope is given by the interval-valued function $B_\alpha(r) = [T_{(k_\alpha)}(r), T_{(n+1-k_\alpha+1)}(r)]$ for critical rank
\begin{equation*}
    k_\alpha = {\rm max}\Big\{k~:~\sum_{i=1}^{n+1} \mathbbm{1}\Big[\sum_{j=1}^{n+1} \mathbbm{1}[\hats^{\rm erl}_j \preceq \hats^{\rm erl}_i] < k\Big] \leq \alpha(n+1)\Big\}.
\end{equation*}
This is also called the $k_\alpha$-th rank envelope.
As Theorem 1 of \cite{Myllymaki2017:Global} expresses, the decision made by the \pval{} of~\eqref{eq:get_pvalue} corresponds with rejecting the hypothesis when $T_{n+1}(r)$ falls outside $B_\alpha(r)$ for any $r\in \Idisc$.

Note that other functional ranking measures have been proposed, for instance the continuous rank and area measures of \cite{Mrkvicka2022:New}. Moreover, \cite{Mrkvicka2023:FDR} proposed a \gls{get} based on the \gls{fdr} criterion for hypothesis tests of a single point pattern. These \gls{fdr} global envelopes are constructed to control the \gls{fdr} across $r\in \Idisc$ for the functional summary statistics.

Multiple testing of replicated point patterns have also been considered in the literature on global envelope testing but only for a global null hypothesis, that is to test the null hypothesis $H_0: P_j = P_0$, for all $j=1,\dots,m$. 
\cite{Mrkvicka2017:Multiple} reports a result of testing if all point patterns in a group have the property of complete spatial randomness, which is done by concatenating the functional summary statistics and doing a \gls{get} on this extended domain. This test has the graphical interpretation allowing the practitioner to visually interpret what causes a potential rejection. However, the testing scenario considered in this work differs substantially from the aforementioned scenario as we test local hypotheses, meaning $H_0^j : P_j = P_0$, thereby making statements for the individual point patterns in a group, while ensuring rigorous control of the {\gls{fdr}} or {\gls{fwer}}. As described in~Section~\ref{sec:methodology} and seen in Figure~\ref{fig:graphical_sweat_gland}, we also have a graphical interpretation. A comparison with the testing procedure in \cite{Mrkvicka2017:Multiple}, presented in Section~\ref{sec:global_test_appendix} of the supplementary material, reveals a significant gain in power when using conformal p-values.

In a multiple Monte Carlo test context with test points $\Xbs_{n+1}, \ldots, \Xbs_{n+m}$, e.g.~multiple \gls{get}, we need to generate $m$ null samples $\Dnull_j$, each of them with $n_j$ observations. Then, for $j=1,\ldots,m$ and $i=1,\ldots,n_j$, we use test point $\Xbs_{n+i}$ and $\Dnull_j$ to compute the conformity score as in~\eqref{eq:ERL conformity score}. Then, the global extreme rank envelope test \pval{}, or equivalently the conformal \pval{}, is
\begin{equation*}
    \hat{p}_j = \frac{1 + \sum_{i=1}^{n_j} \mathbbm{1}[ \hats^{\rm erl}_{j,i} \preceq \hats^{\rm erl}_{j,j}]}{n_j+1}
\end{equation*}
where $\hats^{\rm erl}_{j,i}$ is as in~\eqref{eq:ERL conformity score} but computed with the null sample $\Dnull_j$. In this way, the \pvals{} are independent so that the \gls{bh} procedure in Definition~\ref{def:BH} controls the \gls{fdr}, see~\cite{Benjamini1995:FDR}. However, this can be unfeasible in practice since we need to simulate $\sum_{j=1}^m n_j$ point patterns under $P_0$. Such a procedure with $n_1=\cdots=n_m$ is a naïve \gls{mmctest}.
We notice that multiple Monte Carlo testing is, up to a change in a notation, a specific instance of conformal novelty detection \citep{Bates2023:Testing} where we use independent null samples for each test point.

\subsection{Conformal Multiple Monte Carlo test}\label{sec:methodology}
We consider then the novelty detection setting of Section~\ref{sec:introduction}, where the $\Xbs_i$ are point processes. We define a conformal score with \gls{erl} and assign to each $\Xbs_i$, for $i=1,\ldots, n+m$ the conformal score
\begin{equation}\label{eq:extreme_rank_length_statistic_joint}
    \hats^{\rm joi\text{-}erl}_i = \hats^{\rm erl}(\Xbs_i, (\Xbs_1, \dots, \Xbs_{n}, \Xbs_{n+1}, \ldots \Xbs_{n+m})) = {\bs{u}}_i^{\mathcal{T}_{1:n+m}},
\end{equation}
where ${\bs{u}}_i^{\mathcal{T}_{1:n+m}}$ is the \gls{erl} measure among $\mathcal{T}_{1:n+m} = \{T_1,\dots,T_{n+m}\}$. We will refer to this as the joint \gls{erl} conformal score. {The corresponding conformal p-value for the $j$-th test point $\Xbs_{n+j}$, $j=1,\dots,m$, is
\begin{equation}\label{eq:pval_joi-erl}
    \hat{p}^{\rm joi\text{-}erl}_j = \frac{1}{n+1} \Big(1+\sum_{i=1}^{n} \mathbbm{1}[\hats^{\rm joi\text{-}erl}_i \preceq \hats^{\rm joi\text{-}erl}_{n+j}]\Big).
\end{equation}}%
As an alternative, we consider also a parallel \gls{erl} conformal score, in which the functional ranking is done in parallel for each test point. This has a higher computational cost but can benefit from less dependence between \pvals{}, yielding a more powerful test. This parallel \gls{erl} conformal score is defined as
\begin{equation}\label{eq:extreme_rank_length_statistic_parallel}
    \hats^{\rm par\text{-}erl}_{{j,i}} = \hats^{\rm erl}(\Xbs_i, (\Xbs_1, \dots, \Xbs_{n}, \Xbs_{n+j})) = {\bs{u}}_i^{\mathcal{T}_{1:n,j}},
\end{equation}
for $i\in\{1,\dots,n,n+j\}$ where ${\bs{u}}_i^{\mathcal{T}_{1:n,j}}$ is the \gls{erl} measure among $\mathcal{T}_{1:n,j} = \{T_1,\dots,T_n, T_{n+j}\}$. {Accordingly, the conformal p-value for the $j$-th test point $\Xbs_{n+j}$, $j=1,\dots,m$, is
\begin{equation}\label{eq:pval_par-erl}
    \hat{p}^{\rm par\text{-}erl}_j = \frac{1}{n+1} \Big(1+\sum_{i=1}^{n} \mathbbm{1}[\hats^{\rm par\text{-}erl}_{{j,i}} \preceq \hats^{\rm par\text{-}erl}_{j,n+j}]\Big).
\end{equation}}%
Unfortunately, ties are possible with the \gls{erl}, which would violate Assumption~\ref{A:notie}. But, as noticed in~\cite{Myllymaki2017:Global}, as $n$ grows, this becomes increasingly unlikely. Thus, we assume that there is, in practice, approximately no ties so that Assumption~\ref{A:notie} holds. Moreover, these conformal scores verify Assumption~\ref{A:invariance}. Finally, Assumption~\ref{A:exchangeability} holds as $\Xbs_1,\dots,\Xbs_n$ are sampled \gls{iid} from $P_0$ and independently of $\Xbs_{n+1},\dots,\Xbs_{n+m}$. Therefore, as Assumptions~\ref{A:invariance}-\ref{A:notie} hold, we can apply Theorem~\ref{theorem:conformal_Storey} to control the \gls{fdr} in a novelty detection setting with the conformal \pvals{} {of \eqref{eq:pval_joi-erl} and \eqref{eq:pval_par-erl}}.

Note that, contrary to the \pvals{} defined in Section~\ref{sec:pp_envtest} for global rank envelopes, we require only one null sample $\Dnull$. The conformal \pvals{} are not independent but the validity of Storey's \gls{bh} procedure is guaranteed by Theorem~\ref{theorem:conformal_Storey}. With the same arguments, we conjecture that we can apply the Hochberg procedure to control the \gls{fwer}, as discussed in Section~\ref{sec:controlling_FWER}. We name these procedures \gls{cmmctest}.

As mentioned earlier, one of the attractive properties of \gls{get} is the graphical interpretation, and this can also be done with our proposed multiple testing procedure. In this case the coverage region is adjusted for each test point according to the threshold which the \pval{} is compared to. For Storey's \gls{bh} procedure, the sequence of thresholds is $t_j = j\alpha/\hat{m}_0$ where $\alpha$ is the significance level and $\hat{m}_0$ is the estimate of the number of true nulls in the test set. This sequence of thresholds defines a sequence of $100(1 - t_j)~\%$ coverage regions, specifically $B_{t_j}(r) = [T_{(k_{t_j})}(r), T_{(n+\rev{m}-k_{t_j}+1)}(r)]$ for critical rank $k_{t_j}$, giving upper and lower envelopes for the test point with the $j$-th smallest \pval{} in the test set.

In spatial statistics, it is very common, as for the sweat gland data in \cite{Kuronen2021:Point}, that $n=|\Dnull|$ is small and $P_0$ is also unknown. If the data were not point patterns, bootstrap resampling techniques could be used to augment the null training sample, as in \cite{Guo2008:Adaptive}. One way to address the issue of a small null sample with unknown distribution is to assume a parametric distribution $P(\xbs|\bs{\theta})$ for the point patterns given some parameter $\bs{\theta}$ such that $P_0 = P(\xbs, \bs{\theta}) = P(\xbs|\bs{\theta}) P(\bs{\theta})$. Thus, for each point pattern in $\Dnull$, we can compute an estimate $\hat{\theta}_i$ of $\theta$, for $i=1,\ldots, {n}$. Then, we estimate $P_0$ by $\hat{P}_0 = \hat{P}(\xbs, \bs{\theta}) = \sum_{i=1}^n P(\xbs|\hat{\bs{\theta}}_i)/n$. Finally, we simulate several point patterns, under $\hat{P}_0$, to augment $\Dnull$. Note that, if $m$ is moderately large, this remains way easier than to simulate independent calibration sets as in Section~\ref{sec:pp_envtest},

There are of course multiple downsides to the aforementioned parametric method: (i) model misspecification can occur in the sense that $P_0\neq P(\xbs|\bs{\theta}) P(\bs{\theta})$; (ii) the model fitting can in some cases be computationally complex; (iii) the observed point patterns $\bs{x}_1, \dots, \bs{x}_n$ can be on a small observation window and/or with few points leading to an inaccurate model fit; and (iv) the estimation method may not be consistent.

Alternatively, it can be beneficial to consider a non-parametric technique for sampling under the null. One such technique is to do bootstrap resampling under the empirical distribution of the functional summary statistic in question, as in \cite{Diggle1991:Analysis}. Another way is to use $p$-thinning \citep{Cronie2023:Cross}: some functional summary statistics, including the pair-correlation function, the Ripley $K$-function, and the Besag $L$-function, are theoretically unchanged when doing $p$-thinning, and hence, such a resampling technique can be considered to yield conditionally independent samples from the summary statistic, conditioned on the observed point pattern. However, due to the conditioning on the observed point patterns, such a sample can in general not be considered an exchangeable null sample. We leave investigation on non-parametric techniques as a topic for future research.

\section{Simulation Study}\label{sec:simulation}

In this section, we validate through numerical experiments that the proposed method is more powerful than existing techniques. As a baseline we consider the naïve \gls{mmctest} method which divides the null sample evenly between the test points, thereby reducing the size of the null sample for each individual test point, however by this construction generating independent \pvals{}. {Thus, the only difference between CMMCTest and naïve MMCTest is that, for a fixed simulation budget of $n$, and $m$ test points, \gls{cmmctest} uses the same $n$ simulated test statistics for each \pval{}, and naïve \gls{mmctest} uses $n/m$ simulated test statistics for each \pval{}, different for each test.}

In the simulation study, we consider three types of point processes on the unit square window all with approximate intensity of $200$.
\begin{itemize}
    \item ${\rm Poisson}(200)$: a Poisson process with intensity $200$, modeling complete spatial randomness;
    \item ${\rm Strauss}(250, 0.6, 0.03)$: a Strauss process with intensity parameter $250$, interaction parameter $0.6$, and interaction radius $0.03$, modeling inhibition;
    \item ${\rm LGCP}(5, 0.6, 0.05)$: a log-Gaussian Cox process with mean $5$ and exponential covariance function with variance $0.6$ and scale $0.05$, modelling aggregation.
\end{itemize}
Example realizations of the point processes are shown in Figure~\ref{fig:realizations}.
The code used for the study can be found in the \href{https://github.com/Martin497/Conformal-novelty-detection-for-replicate-point-patterns}{github repository}\footnote{\url{https://github.com/Martin497/Conformal-novelty-detection-for-replicate-point-patterns}}.

In the following, we begin by showing the performance gains of the \gls{cmmctest} method over the \gls{mmctest} method when controlling the \gls{fdr} using Storey's \gls{bh} procedure, assuming either that the null distribution is known, or a setting where the null distribution is estimated from a number of training point patterns. This is followed by a study of the performance of \gls{cmmctest} for different choices of conformal scores, and afterwards we observe how the significance level, $\alpha$, the number of tests, $m$, and the size of the null sample, $n$, influences the power. Finally, we present results from the same simulation study, but controlling the \gls{fwer} with the Hochberg procedure. In all cases, we estimate the \gls{fdr} (or \gls{fwer}) and \gls{tdr} with $2000$ simulations.

\begin{figure}[t]
    \centering
    \begin{minipage}{0.3\linewidth}
        \centering
        \subfloat[${\rm Strauss}(250, 0.6, 0.03)$]{\includegraphics[width=1\linewidth, page=11]{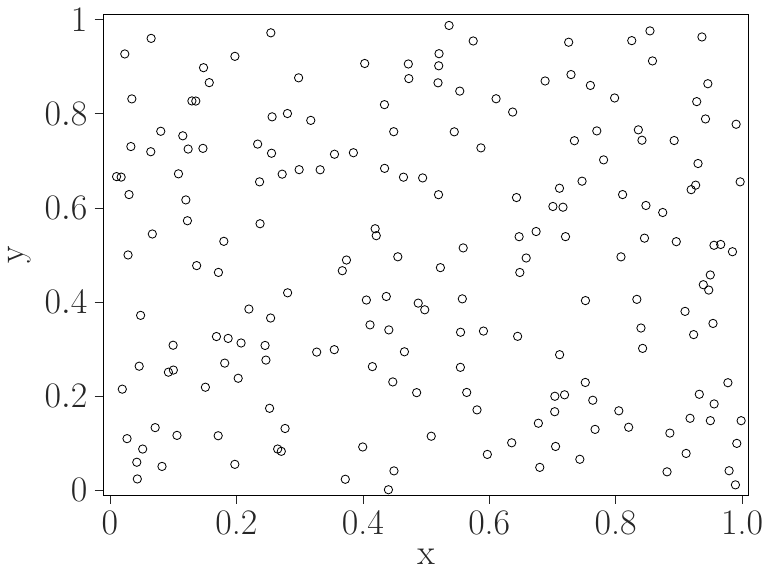}\label{subfig:realization_Strauss}}
    \end{minipage}~
    \begin{minipage}{0.3\linewidth}
        \centering
        \subfloat[${\rm Poisson}(200)$]{\includegraphics[width=01\linewidth, page=12]{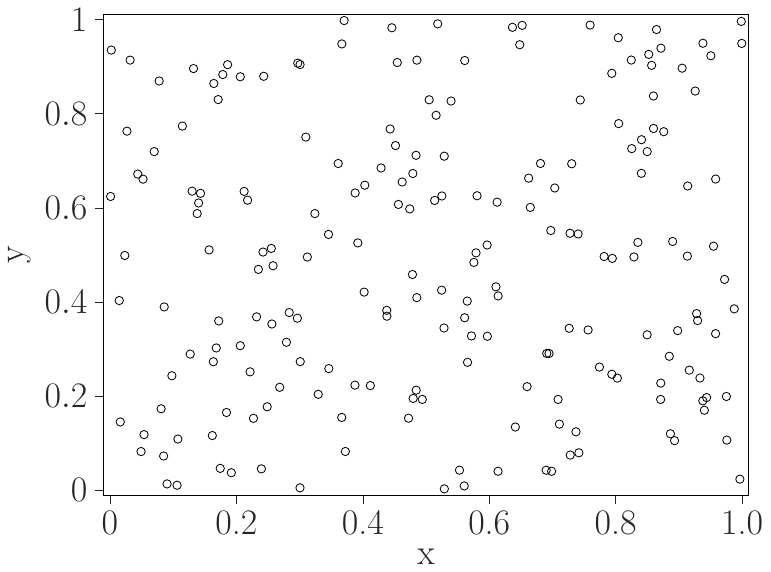}\label{subfig:realization_Poisson}}
    \end{minipage}~
    \begin{minipage}{0.3\linewidth}
        \centering
        \subfloat[${\rm LGCP}(5, 0.6, 0.05)$]{\includegraphics[width=1\linewidth, page=13]{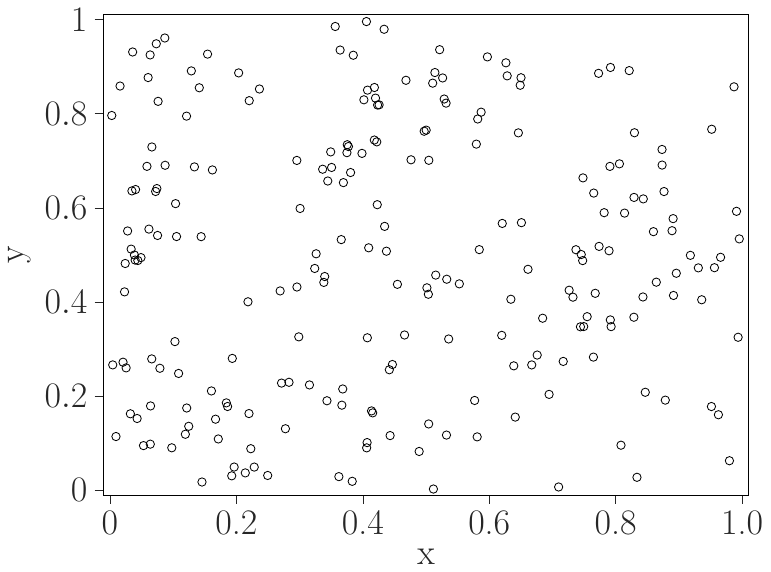}\label{subfig:realization_LGCP}}
    \end{minipage}
    \caption{Example realizations of the point processes considered in the simulation study.}
    \label{fig:realizations}
\end{figure}

\subsection{Power comparison: Storey's BH procedure}
\begin{figure}[t]
    \centering
    \includegraphics[width=0.9\linewidth, page=1]{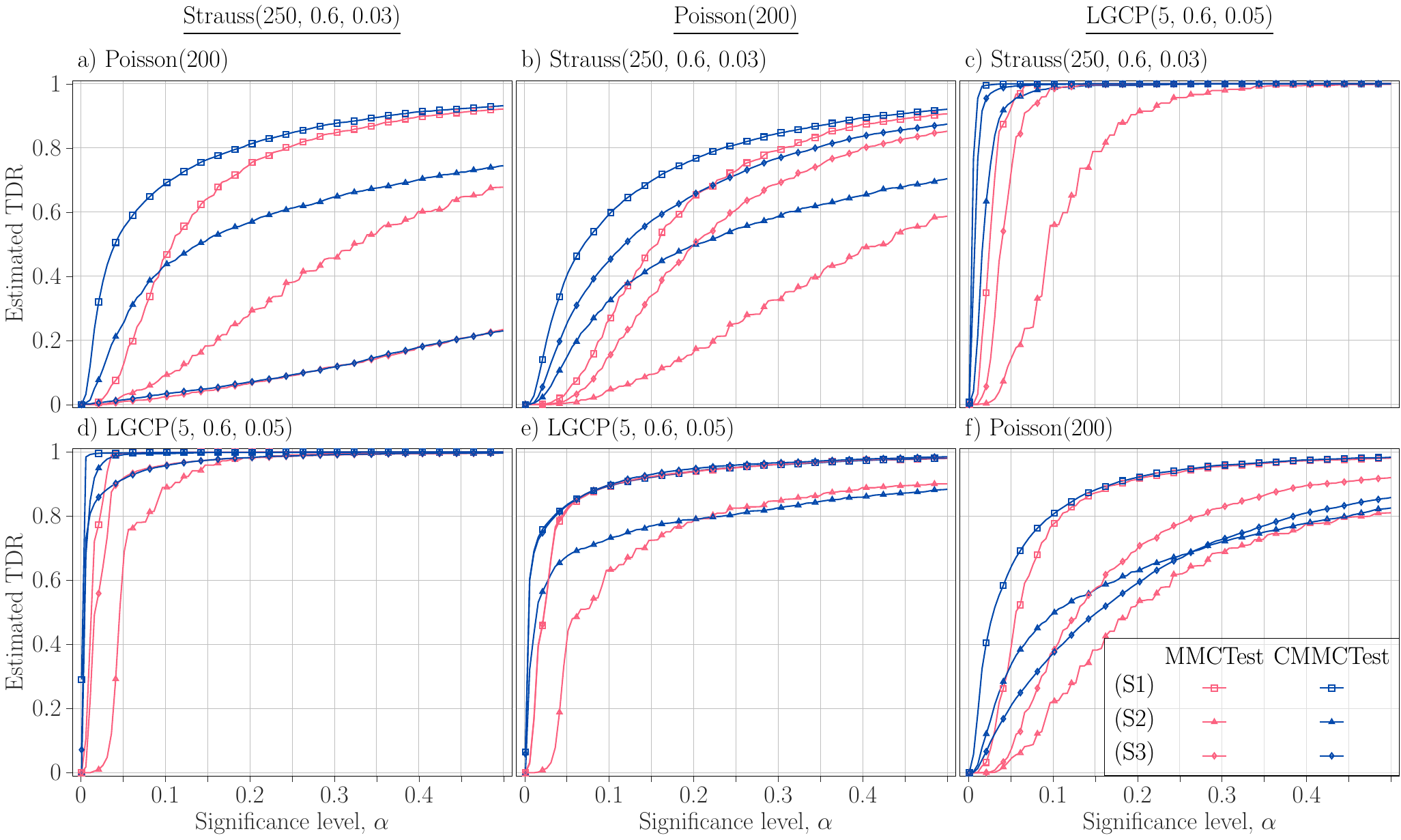}
    \vspace{-2mm}
    \caption{Estimated power curves against the significance level, $\alpha$, when using Storey's \gls{bh} procedure with $\lambda=0.5$. At the top of the figure, the null distribution is underlined, while the non-true null distribution is read above each individual plot. All the plots share the same legend.}
    \label{fig:power_curves}
\end{figure}

\begin{figure}[t]
    \centering
    \includegraphics[width=0.95\linewidth, page=2]{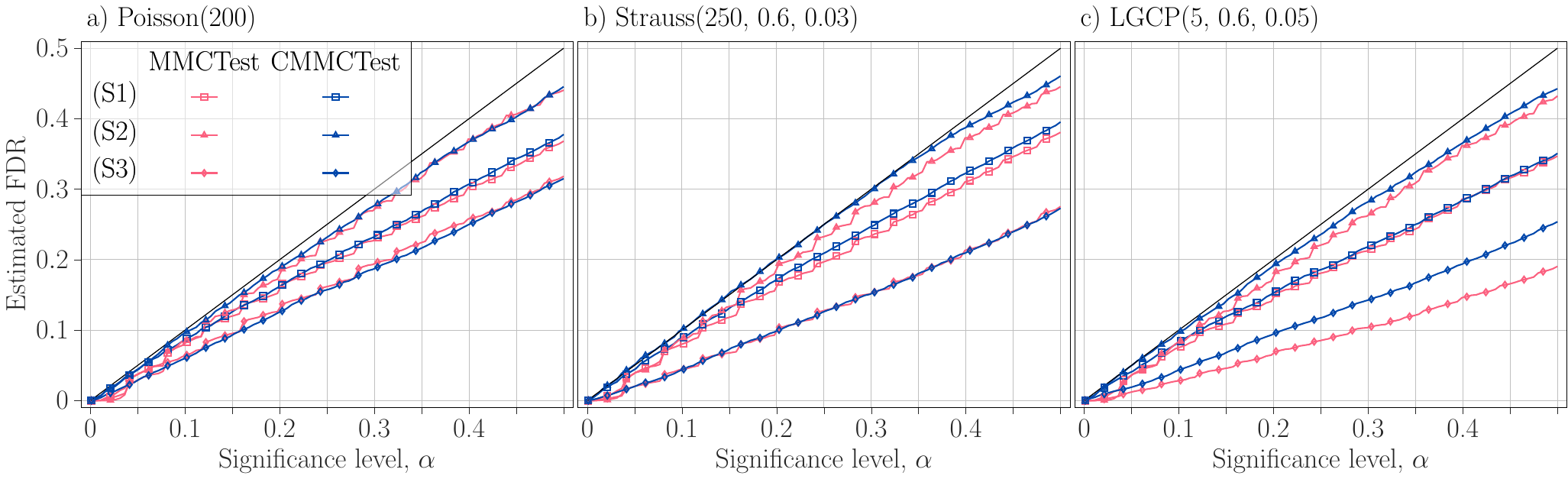}
    \vspace{-2mm}
    \caption{Estimated \gls{fdr} curves against the significance level, $\alpha$, when using Storey's \gls{bh} procedure with $\lambda=0.5$. The black straight line, ${\rm FDR}=\alpha$, shows the nominal level. All the plots share the same legend.}
    \label{fig:FDR_curves}
\end{figure}
We begin by comparing \gls{mmctest} to \gls{cmmctest} using power and \gls{fdr} curves. In this simulation study, the centered L-function \citep{Baddeley2015:Spatial} is used as the functional summary statistic, and the function ranking is done using the \gls{erl} measure in the \textit{R} package \textit{GET} \citep{Myllymaki2023:Global}, doing parallel ranking of the test point patterns,~\eqref{eq:extreme_rank_length_statistic_parallel}. We consider three cases:
\begin{enumerate}[label=({S\arabic*})]
    \item known null distribution with $m_0=5$;\label{enum:S1}
    \item known null distribution with $m_0=9$;\label{enum:S2}
    \item unknown null distribution with $m_0=5$ and $10$ observations from the null.\label{enum:S3}
\end{enumerate}
In all these cases, we consider a simulation budget of $n=2500$, and $m=10$ test points. {Hence, for CMMCTest the null sample consists of $2500$ point patterns while for MMCTest the null sample for each test point pattern consists of just $250$ point patterns.} When fitting the models in Scenario~\ref{enum:S3} we use maximum likelihood for Poisson, maximum profile pseudolikelihood for Strauss, and minimum contrast estimation with Ripley's K-function for the log-Gaussian Cox process. This is done with default parameters using the standard model fitting functionality in \textit{spatstat} \citep{Baddeley2015:Spatial}.

In Figure~\ref{fig:power_curves}, where all the plots share the same legend, we show estimated power curves against the significance level, $\alpha$, when using Storey's \gls{bh} procedure with $\lambda=0.5$ for each of the $6$ combinations of null distributions and non-true null distributions and in the three scenarios \ref{enum:S1}, \ref{enum:S2}, and \ref{enum:S3}.
For the same cases, we show in Figure~\ref{fig:FDR_curves} the estimated \gls{fdr} curves.
We observe from Figure~\ref{fig:power_curves} that in nearly all cases \gls{cmmctest} dominates \gls{mmctest} in terms of power, with a few exceptions where the power is almost equal at relative large significance levels. Moreover, higher power is achieved when $m_0=5$ compared to when $m_0=9$, which is a standard consequence of multiplicity control: recall Storey's BH procedure from Section~\ref{sec:multiplehyptesting} \citep{Benjamini1995:FDR,Storey2004:Strong}. Also, approximating the null distribution from only $10$ samples results in notable power loss in most cases, particularly in Figure~\ref{fig:power_curves}a in which the null distribution is ${\rm Strauss}(250, 0.6, 0.03)$ and the non-true null distribution is ${\rm Poisson}(200)$. The power loss is substantial here since the considered Strauss and Poisson processes can yield quite similar point patterns, and in this specific numerical experiment, the estimate of the interaction parameter for one of the Strauss point patterns is one, thereby coinciding with Poisson.
On the other hand, Figure~\ref{fig:power_curves}e shows no loss in power due to approximating the null, which we attribute to the lack of importance of the exact value of the intensity in the L-function.
Figure~\ref{fig:FDR_curves} shows that \gls{cmmctest} is conservative when using Storey's \gls{bh} procedure. This conservativeness is exacerbated when $m_0$ is smaller, and even more so when the null distribution must be approximated.
\begin{figure}[t]
    \centering
    \includegraphics[width=0.95\linewidth, page=3]{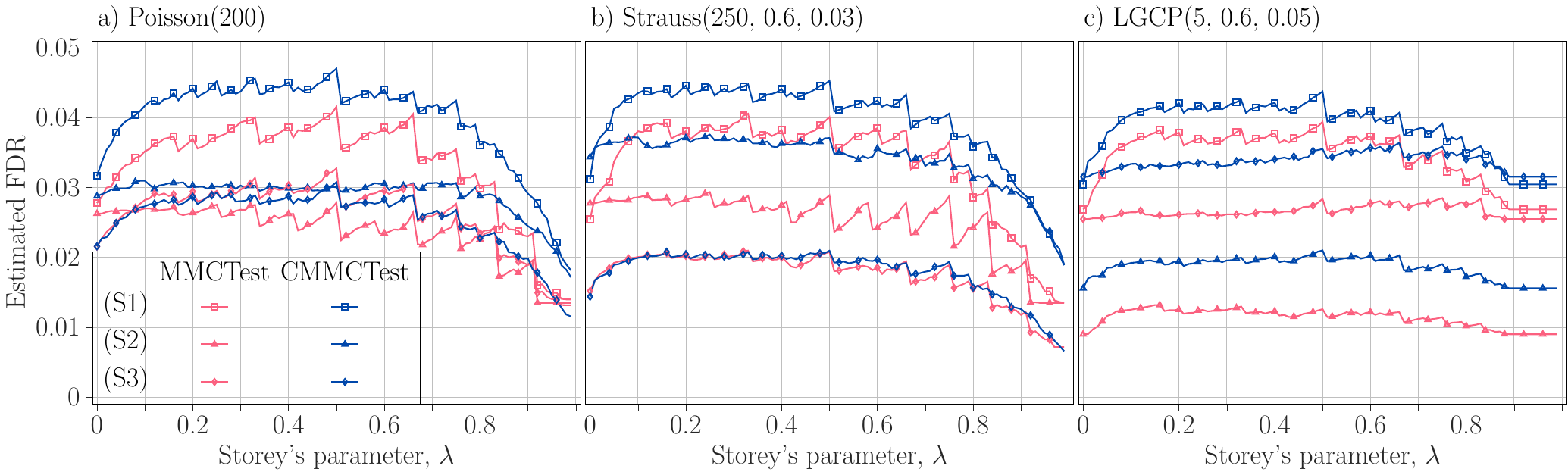}
    \vspace{-2mm}
    \caption{Estimated \gls{fdr} according to the choice of Storey's parameter, $\lambda$. All the plots share the same legend.}
    \label{fig:Storey_lambda}
\end{figure}

The preceding numerical experiments raises a question regarding the choice of the $\lambda$ hyperparameter used in Storey's estimator, cf.~Definition~\ref{def:storey}. We have explored the influence of this choice for the same scenarios as previous for $\alpha=0.05$. We see in Figure~\ref{fig:Storey_lambda} that Storey's \gls{bh} procedure is not sensitive to the choice of the parameter $\lambda$, and that the procedure is conservative for any $\lambda \in (0,1)$ as suggested by Theorem~\ref{theorem:conformal_Storey}. Particularly, it seems that choosing $\lambda \in [0.1, 0.7]$ yields comparable \gls{fdr}. We skip showing the corresponding power curves as they output the same conclusions.

\subsection{Choice of conformal score}\label{subsec:choice_of_score}
So far we have only considered one type of conformal score in using the \gls{erl} measure for the centered L-function doing the ranking in parallel for each of the test point patterns. We consider now different GET measures, specifically the continuous rank measure, denoted \textit{cont}, and the area measure, denoted \textit{area}, both proposed by \cite{Mrkvicka2022:New}. {Notably, both these scores are continuous solving any theoretical problems related to ties.} Moreover, we use also the J-function (see \cite{Lieshout1996:Nonparametric}) in place of the centered L-function. Finally, we try also the joint ranking of~\eqref{eq:extreme_rank_length_statistic_joint}.
\begin{figure}[t]
    \centering
    \includegraphics[width=0.86\linewidth, page=7]{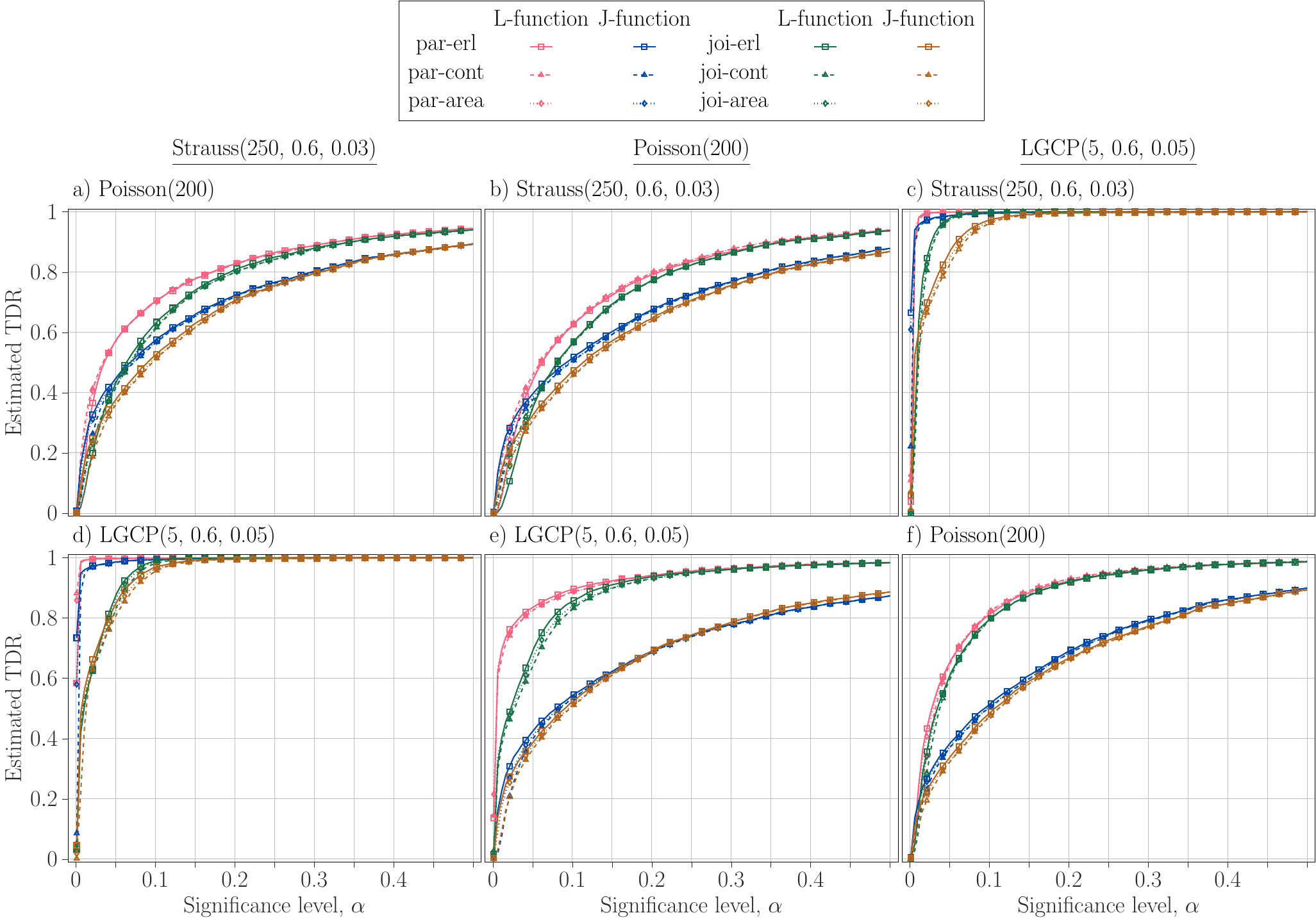}
    \vspace{-2mm}
    \caption{Estimated power curves against the significance level, $\alpha$, when using Storey's \gls{bh} procedure with $\lambda=0.5$ for a variety of conformal scores. At the top of the figure, the null distribution is underlined, while the non-true null distribution is read above each individual plot.}
    \label{fig:TDR_scores}
\end{figure}
\begin{figure}[t]
    \centering
    \includegraphics[width=0.98\linewidth, page=8]{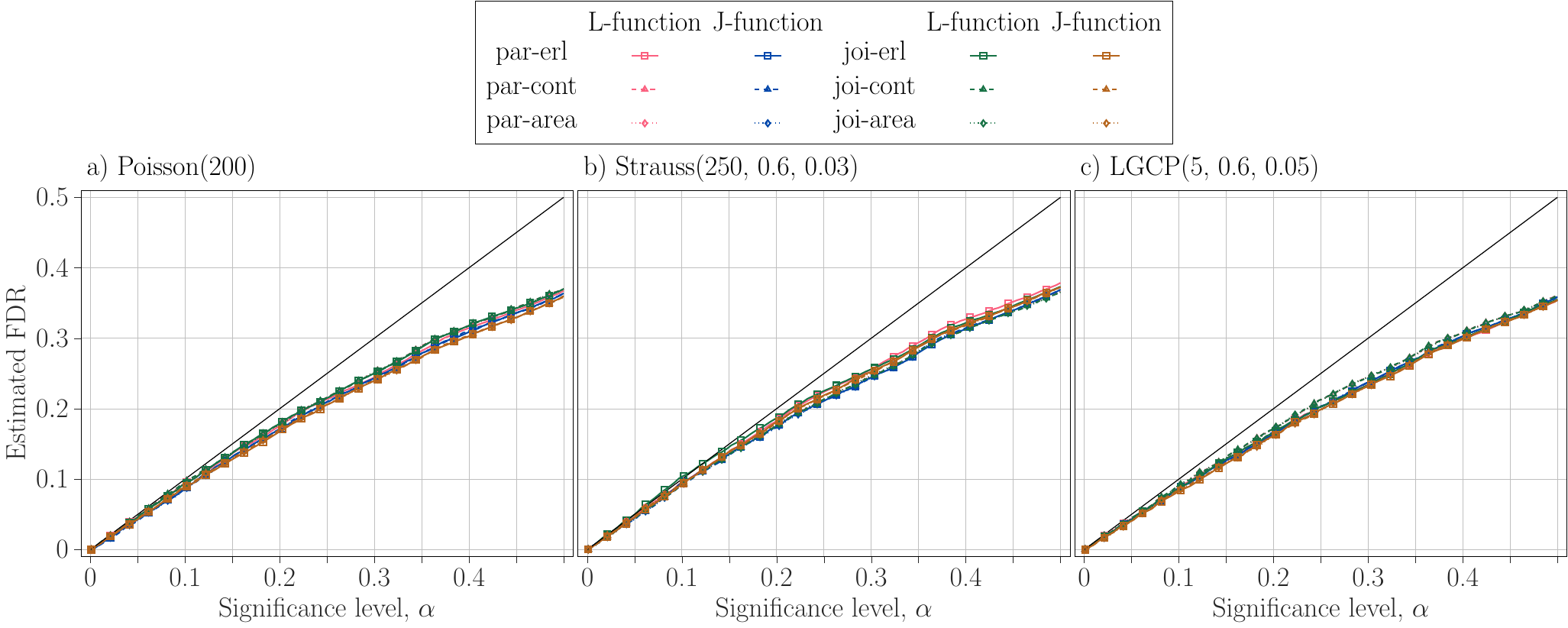}
    \vspace{-2mm}
    \caption{Estimated \gls{fdr} curves against the significance level, $\alpha$, when using Storey's \gls{bh} procedure with $\lambda=0.5$ for a variety of conformal scores. The black straight line, ${\rm FDR}=\alpha$, shows the nominal level.}
    \label{fig:FDR_scores}
\end{figure}

We use Storey's \gls{bh} procedure with $\lambda=0.5$, $n=2500$, and consider $m=10$ with $m_0=5$. Figure~\ref{fig:TDR_scores} shows the \gls{tdr} for all the test cases, while the corresponding \gls{fdr} is found in Figure~\ref{fig:FDR_scores}. We observe in Figure~\ref{fig:TDR_scores} that no substantial differences occur whether we use the \gls{erl} measure, the \textit{cont} measure, or the \textit{area} measure. Moreover, we notice that the centered L-function tends to yield a more powerful test than that of the J-function, which is consistent with the observations of \cite{Mrkvicka2017:Multiple}.

Finally, we see that doing parallel ranking of the test point patterns in all cases yields a uniformly more powerful test than doing joint ranking. This is due to an additional dependency added in the \pvals{} when doing the joint ranking. Hence, there is a trade-off between computational complexity and performance: if the cost of doing the functional ranking is not substantial, the parallel ranking is preferred, however, in cases where $m$ is very large so the computational demands of doing the parallel ranking is impractical, the joint ranking can be used. We remark also that in most cases, using the \gls{cmmctest} with the joint ranking is still more powerful than \gls{mmctest}, which can be seen by comparing Figure~\ref{fig:TDR_scores} and Figure~\ref{fig:power_curves}.

\subsection{Robustness against multiplicity}
In this scenario, we consider different settings of triplets $(\alpha, m, n)$. This has been studied theoretically by \cite{Mary2022:Semi} and \cite{Marandon2024:Adaptive}, while in the current study we evaluate numerically the power through simulations and focus on the benefit of using the \gls{cmmctest} as compared to \gls{mmctest}. The null distribution is ${\rm Poisson}(200)$ and the non-true null distribution is ${\rm Strauss}(250, 0.6, 0.03)$, setting $m_0=m/2$ in all cases, and we use Storey's \gls{bh} procedure with $\lambda=0.5$. In Figure~\ref{fig:nplot}, we show the \gls{tdr} for $\{(\alpha, m, n) : \alpha\in\{0.05, 0.1, 0.2\}, m\in\{6, 10, 20, 30\}, n\subset\{240, \dots, 2520\}\}$. We observe that for \gls{cmmctest}, the value of $m$ has little influence, and when $\alpha=0.2$ having a larger $m$ can even give slight improvements in \gls{tdr}. Meanwhile, a large $m$ is detrimental to the power of \gls{mmctest}, and only the case $\alpha=0.2$ with $m\in\{6, 10\}$ shows comparable performance to the \gls{cmmctest}.
\begin{figure}[t]
    \centering
    \includegraphics[page=9, width=0.95\linewidth]{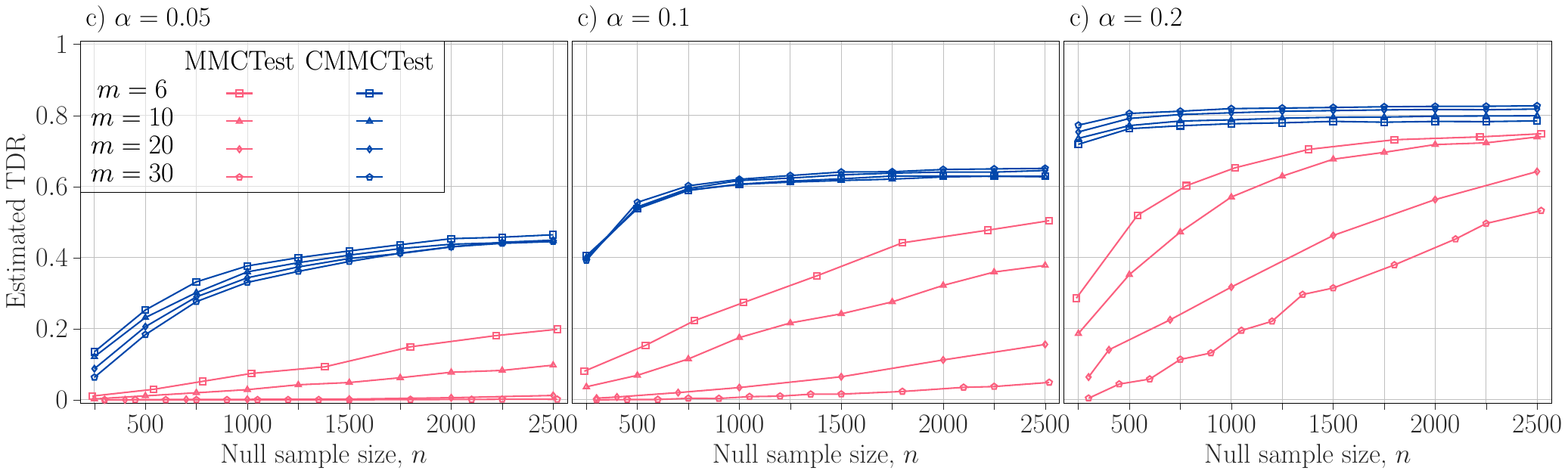}
    \vspace{-2mm}
    \caption{Estimated \gls{tdr} of the \gls{mmctest} and \gls{cmmctest} for varying triplets $(\alpha, m, n)$. All the plots share the same legend.}
    \label{fig:nplot}
\end{figure}

\subsection{Power comparison: Hochberg procedure}\label{subsec:controlling_FWER_results}
As shown in Theorem~\ref{thm:conformal_hochberg}, using the Hochberg procedure we can control the \gls{fwer} with conformal \pvals{}. We report here results using the Hochberg procedure for scenarios \ref{enum:S1}, \ref{enum:S2}, and \ref{enum:S3}: \gls{tdr} curves are shown in Figure~\ref{fig:power_Hoch}, and \gls{fwer} curves are shown in Figure~\ref{fig:FWER_Hoch}. Firstly, we see from Figure~\ref{fig:FWER_Hoch} that indeed we control the \gls{fwer}, and that estimating the null model, or having a larger $m_0$, results in a more conservative test. Inspecting Figure~\ref{fig:power_Hoch} and comparing to Figure~\ref{fig:power_curves} we notice that for $m=10$ with $m_0=5$ the power loss of the Hochberg procedure in comparison to Storey's \gls{bh} procedure is not so pronounced, however, when $m_0=9$ the strict \gls{fwer} criterion leads to a severe loss in power. In all cases, estimating the null model results in a loss in power as expected and also seen previously for Storey's \gls{bh} procedure. Finally, we note once again that the \gls{cmmctest} outperforms the \gls{mmctest} in nearly all considered cases.
\begin{figure}
    \centering
    \includegraphics[width=0.9\linewidth, page=4]{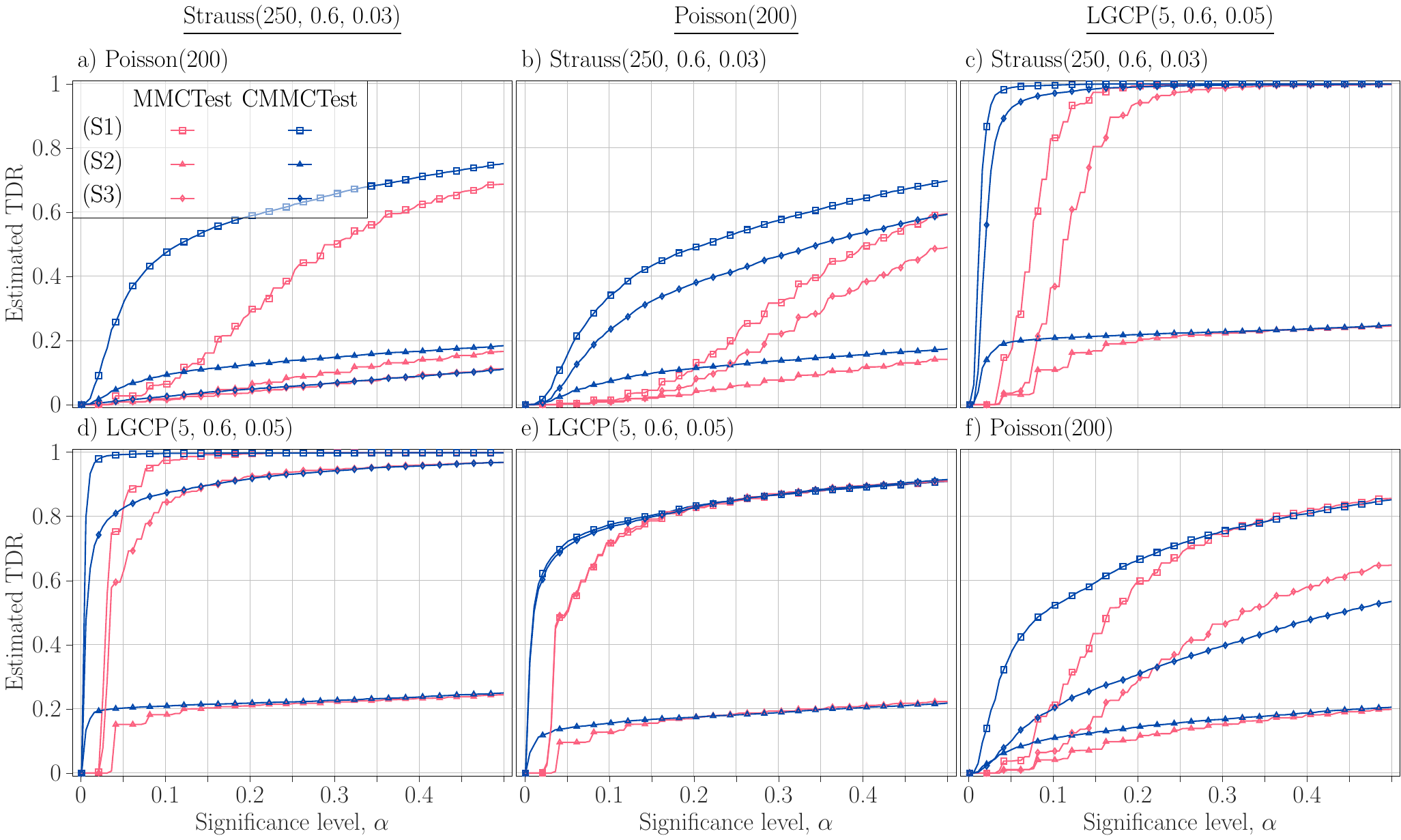}
    \vspace{-2mm}
    \caption{Estimated power curves against the significance level, $\alpha$, when using the Hochberg procedure. At the top of the figure, the null distribution is underlined, while the non-true null distribution is read above each individual plot. All the plots share the same legend.}
    \label{fig:power_Hoch}
\end{figure}

\begin{figure}
    \centering
    \includegraphics[width=0.98\linewidth, page=5]{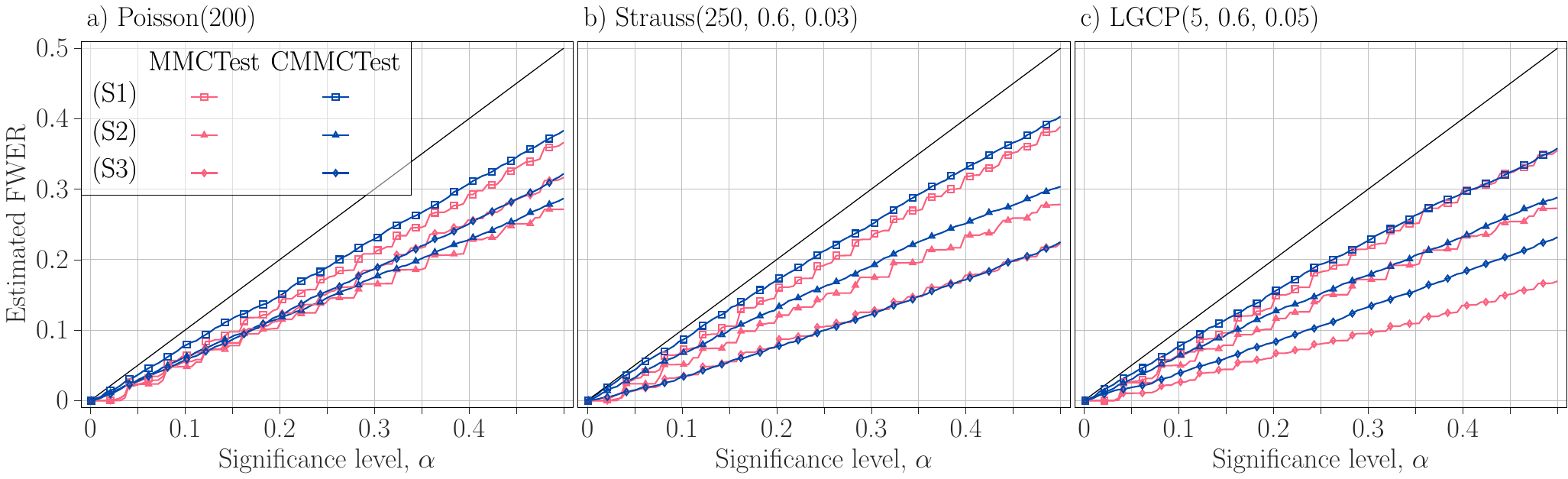}
    \vspace{-2mm}
    \caption{Estimated \gls{fwer} curves against the significance level, $\alpha$, when using the Hochberg procedure. The black straight line, ${\rm FWER}=\alpha$, shows the nominal level. All the plots share the same legend.}
    \label{fig:FWER_Hoch}
\end{figure}

\section{Real data set: Sweat glands}\label{sec:real_data}
The sweat gland data considered in \cite{Kuronen2021:Point} includes three groups of point patterns, see Figure~\ref{fig:sweat_gland_data}: (i) a control group of $4$ subjects (numbers 96, 149, 203, 205)\footnote{We have excluded subject 97 to avoid problems arising from having a different observation window.} (ii) a group of $5$ subjects (numbers 10, 20, 40, 61, 71) suspected to have neuropathy (MNA), and (iii) a group of $5$ subjects (numbers 23, 36, 42, 50, 73) diagnosed with MNA. For the point pattern data, three different models were proposed: (i) a sequential point process model, (ii) a sequential point process model with noise (SMWN), and (iii) a generative model (GM). For each model, a parameter estimation method was proposed, which for the sequential point process models was based on maximum likelihood and for the generative model relied on approximate Bayesian computation Markov chain Monte Carlo sampling since for this model the likelihood is intractable.

\begin{figure}
    \centering
    \includegraphics[width=0.8\linewidth]{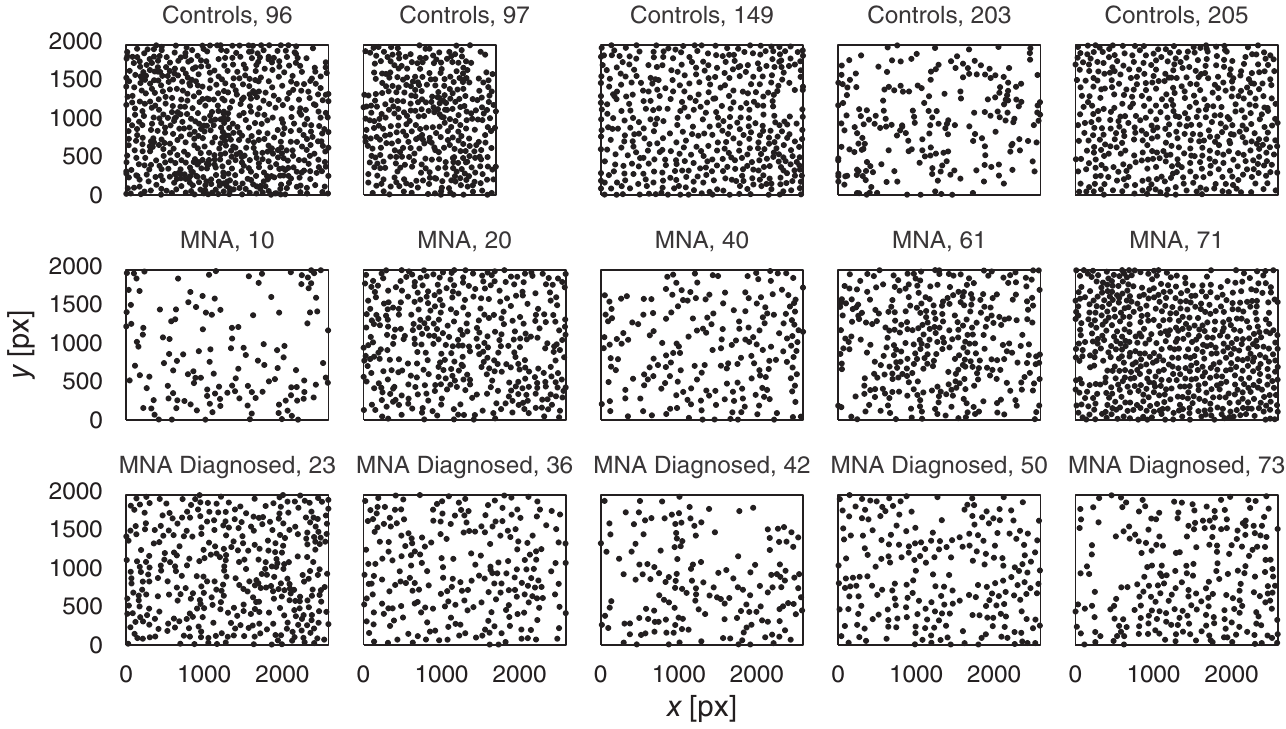}
    \vspace{-4mm}
    \caption{The sweat gland data. Figure taken from \cite{Kuronen2021:Point}.}
    \label{fig:sweat_gland_data}
\end{figure}

With this data different null hypotheses can be considered. Firstly, one could consider the null hypothesis that a subject does not have MNA. In this case, the control group should be considered as the null sample, and we can run a multiple testing procedure on the MNA suspected and MNA diagnosed groups. Secondly, one could consider the null hypothesis that a subject has MNA. In this case, the MNA diagnosed group should be considered as the null sample, and we can run a multiple testing procedure on the MNA suspected and control groups. For both of these null hypotheses, we could use any (or all) of the proposed models and associated inference algorithms proposed in \cite{Kuronen2021:Point} in order to increase the size of the null sample, as $5$ is not sufficient.

Figure~\ref{fig:sweat_gland_rejections} shows the number of rejections in the four cases considering either the MNA diagnosed subjects as the null sample or the control subjects as the null sample, and using either the sequential model with noise or the generative model as the parametric model, discussed in Section~\ref{sec:methodology}, to simulate point patterns. In all cases, Storey's \gls{bh} procedure is used with $\lambda=0.5$ and simulating $n=5000$ point patterns under the fitted null distributions.
\begin{figure}
    \centering
    \includegraphics[width=0.95\linewidth, page=6]{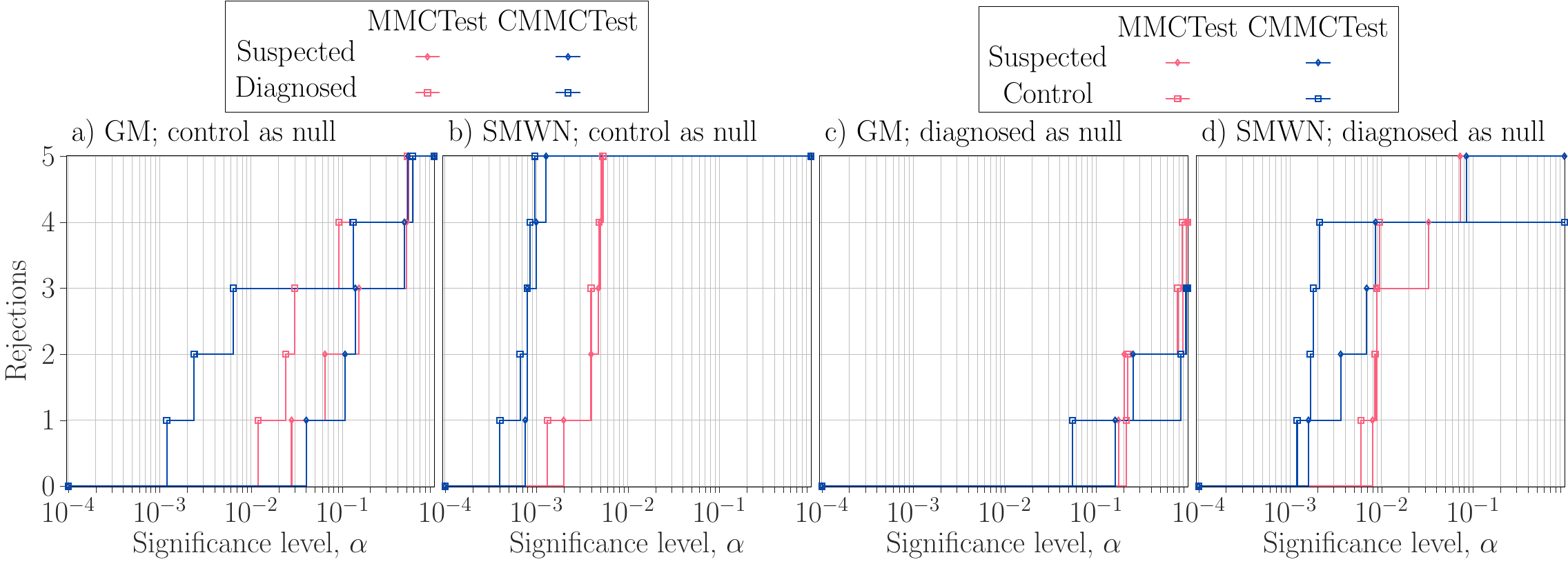}
    \vspace{-2mm}
    \caption{Number of rejections against the significance level, $\alpha$.}
    \label{fig:sweat_gland_rejections}
\end{figure}
We observe from Fig~\ref{fig:sweat_gland_rejections}, that using the sequential model with noise generally leads to more rejections than for the generative model. 
{This may indicate that the sequential model with noise yields a better and more representative fit on the null sample than the generative model. However, given the small sample size, this may also be due to the test statistics variability in different models. Further,} it reveals a profound weakness of the parametric data augmentation in that different hypothesized models can result in very different conclusions made from the test. We also observe more rejections of the control subjects than the MNA suspected subjects when using the MNA diagnosed subjects as the null sample, which is desirable. Moreover, we observe more rejections when using \gls{cmmctest} compared to \gls{mmctest} in almost all cases, indicating the power improvements we expect from the \gls{cmmctest}.

In Figure~\ref{fig:graphical_sweat_gland} we show the global envelopes which gives a graphical interpretation, as discussed in Section~\ref{sec:methodology}, when using the sequential model with noise with the control subjects as the null sample. We use Storey's \gls{bh} procedure with $\lambda=0.5$ at level $\alpha=0.05$, the centered L-function, and the parallel \gls{erl} measure. The observed summary statistics are compared to the upper and lower envelopes with red dots indicating that the observation is more extreme than the null sample at the nominal level. Storey's estimate yields $\hat{m}_0 = 2$.
\begin{figure}[t]
    \centering
    \includegraphics[width=0.8\linewidth]{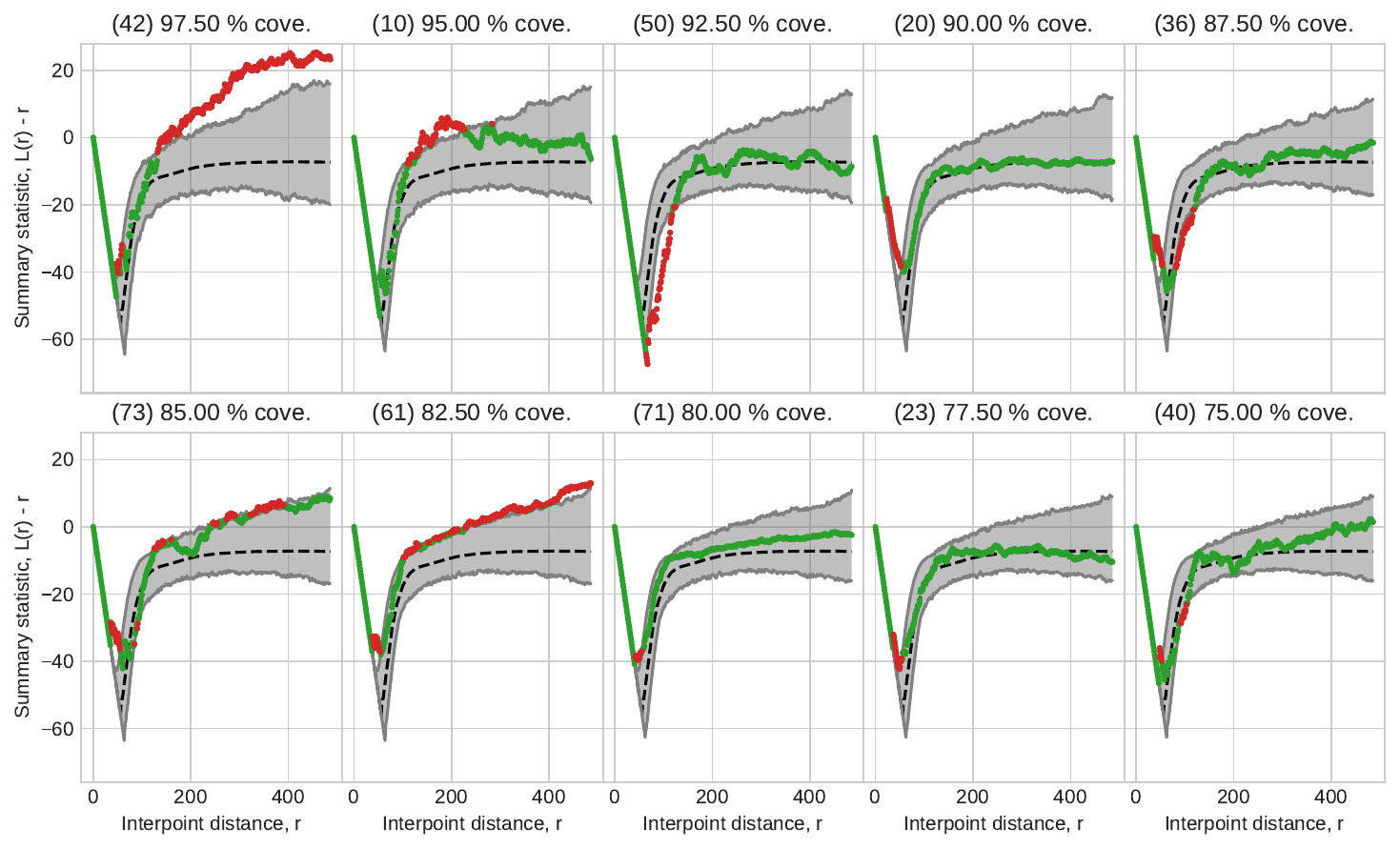}
    \vspace{-3mm}
    \caption{Graphical interpretation of rejection decisions for the sequential model with noise with control as null.}
    \label{fig:graphical_sweat_gland}
\end{figure}

Figure~\ref{fig:graphical_sweat_gland} shows that all the point patterns are rejected meaning that all the MNA suspected and diagnosed subjects are statistically significantly different from the control subjects, consistent with Figure~\ref{fig:sweat_gland_rejections}. The interesting part here is that we can interpret the reason for rejection: for all the point patterns (except subject 10) the summary statistic falls outside of the envelope at small values of $r$, while further some of the point patterns (including subjects 42, 10, and 61) breaks the upper envelope at higher values of $r$. This gives the interpretation that to make the rejection decision, we are exploiting the properties that the L-functions of the control subjects have a lower variability on the level of repulsion at small $r$, while tending to have repulsive/less clustered behavior at larger $r$.

\section{Conclusion}\label{sec:conclusion}
We proposed conformal multiple Monte Carlo testing (CMMCTest), a framework for multiple Monte Carlo testing using recent developments on conformal novelty detection. CMMCTest efficiently re-uses the null sample, cutting computational costs while maintaining rigorous control over \gls{fdr} or \gls{fwer}. We highlight the connection between conformal novelty detection and Barnard's Monte Carlo test, and investigate its use within the area of spatial statistics. Here functional ranking measures can be used to construct conformal scores yielding powerful multiple testing procedures. Simulations and a real-world application to sweat gland data showed CMMCTest's superior statistical power within a fixed computational budget.

Directions for future work include exploring the existence of more powerful multiple testing procedures controlling the \gls{fwer} with conformal p-values and non-parametric data augmentation techniques to avoid the problems with the parametric method.

\section{Acknowledgments}
The authors extend their gratitude to Ottmar Cronie and Etienne Roquain for fruitful discussions and give thanks to Aila Särkkä for discussions surrounding real data applications. \rev{A special mention goes to Pierre Neuvial who provided invaluable insights into conditions for familywise error rate control.}


\newpage
\appendix

\setcounter{figure}{0}
\setcounter{lemma}{0}
\renewcommand{\thelemma}{\Alph{section}.\arabic{lemma}}
\setcounter{theorem}{0}
\renewcommand{\thetheorem}{\Alph{section}.\arabic{theorem}}
\setcounter{definition}{0}
\renewcommand{\thedefinition}{\Alph{section}.\arabic{definition}}
\setcounter{corollary}{0}
\renewcommand{\thecorollary}{\Alph{section}.\arabic{corollary}}

\section{On the distribution of ordered conformal p-values}\label{sec:distribution_appendix}
We present in this section some new results regarding the joint distribution of the order statistics of the conformal p-values. First, we find the joint probability mass function.
\begin{lemma}\label{lem:conformal_joint_order}
    Let assumptions~\ref{A:invariance}-\ref{A:notie} be true for a conformal score $\hats$, and consider $m$ test points $\Xbs_{n+1}, \dots, \Xbs_{n+m}$ such that $(\Xbs_{n+i})_{i\in\mathcal{H}_0}$ are jointly independent of each other. Then, the joint probability mass function of the order statistics of the \pvals{}, denoted $\hat{\bs{p}}_{()}=[\hat{p}_{(1)},\dots,\hat{p}_{(m)}]^\top$, is
    \begin{equation}\label{eq:joint_order}
        \mathbb{P}(\hat{\bs{p}}_{()} = \bs{j}/(n+1)) = \mathbbm{1}[j_1\leq\cdots\leq j_m] \frac{n!m!}{(n+m)!},
    \end{equation}
    for $\bs{j}=[j_1,\dots,j_m]^\top$,
    i.e., the discrete uniform distribution on the integer order set
    \begin{equation*}
        \{(x_1,\dots,x_m):x_1\leq \cdots \leq x_m, x_j \in \{1/(n+1),\dots,1\},j=1,\dots,m\}.
    \end{equation*}
\end{lemma}
\begin{proof}
    We know the joint probability mass function, see~\eqref{eq:joint_mass}, and notice that the joint probability mass function is invariant to permutations. This is exploited to find the joint distribution of the order statistics.
    Specifically, \cite{Song2024:Order} Theorem 2.1.2 gives that
    \begin{equation*}
        \mathbb{P}(\hat{\bs{p}}_{()} = \bs{j}/(n+1)) = \frac{\mathbbm{1}[j_1\leq\cdots\leq j_m]}{{\prod_{k=1}^{n+1} M_k(\bs{j})!}} \displaystyle\sum_{\bs{i}\in\mathcal{A}^m} \mathbb{P}\bigg(\bigcap_{k=1}^m\hat{p}_{i_k} = \frac{j_k}{n+1}\bigg),
    \end{equation*}
    where $\mathcal{A}^m$ is the set of arrangements of $\{1, \dots, m\}$.
    Noticing that the joint probability mass function is invariant to permutations, the sum can be replaced by a multiplication of $|\mathcal{A}^m| = m!$. Inserting the joint probability mass function, see~\eqref{eq:joint_mass}, yields~\eqref{eq:joint_order}.
\end{proof}
Below is a technical lemma which will be required later.
\begin{lemma}\label{lem:sum_rising_factorial}
    For $n \geq 1$ and $m \geq 0$ the following identity holds
    \begin{equation*}
        \sum_{y=1}^{n} \frac{(y+m-1)!}{(y-1)!} = \frac{(n+m)!}{(m+1)(n-1)!}.
    \end{equation*}
\end{lemma}
\begin{proof}
    Dividing by $m!$ on the left hand side yields
    \begin{equation*}
        \sum_{y=1}^{n} \frac{(y+m-1)!}{(y-1)!m!} = \sum_{y=1}^{n} \binom{y-1+m}{m} = \binom{n+m}{m+1} = \frac{1}{m!} \frac{(n+m)!}{(m+1)(n-1)!}
    \end{equation*}
    where the second equality is the Hockey Stick Identity, see \cite{Jones1996:Hockey}.
\end{proof}
Equipped with the joint distribution, we can also derive the marginals.
\begin{corollary}\label{cor:conformal_marginal_order}
    Let assumptions~\ref{A:invariance}-\ref{A:notie} be true for a conformal score $\hats$, and consider $m$ test points $\Xbs_{n+1}, \dots, \Xbs_{n+m}$ such that $(\Xbs_{n+i})_{i\in\mathcal{H}_0}$ are jointly independent of each other. Then, the marginal probability mass function of the $i$-th smallest p-value is
    \begin{align*}
        \mathbb{P}(\hat{p}_{(i)}=j/(n+1)) 
        &= \frac{n!m!}{(n+m)!} \binom{j+i-2}{i-1}\binom{n+m-j-i+1}{\rev{m}-i},
    \end{align*}
    {which is the probability mass function of the negative hypergeometric distribution.}
\end{corollary}
\begin{proof}
    We compute directly the marginal by summing over the joint:
    \begin{equation*}
        \mathbb{P}(\hat{p}_{(i)}=j/(n+1)) = \frac{n!m!}{(n+m)!} \bigg(\sum_{y_{i-1}=1}^{j}\cdots\sum_{y_2=1}^{y_3}\sum_{y_1=1}^{y_2} 1\bigg) \bigg(\sum_{y_{m}=j}^{n+1}\cdots\sum_{y_{i+2}=j}^{y_{i+3}}\sum_{y_{i+1}=j}^{y_{i+2}} 1\bigg).
    \end{equation*}
    The result follows immediately by evaluating the sums. Specifically, computing the first sums
    \begin{equation*}
        \sum_{y_{i-1}=1}^{j}\cdots\sum_{y_2=1}^{y_3}\sum_{y_1=1}^{y_2} 1 = \sum_{y_{i-1}=1}^{j}\cdots\sum_{y_2=1}^{y_3} y_2 = \sum_{y_{i-1}=1}^{j}\cdots\sum_{y_2=1}^{y_3} \frac{y_2!}{(y_2-1)!}.
    \end{equation*}
    Applying Lemma~\ref{lem:sum_rising_factorial} repeatedly gives
    \begin{equation*}
        \sum_{y_{i-1}=1}^{j}\cdots\sum_{y_2=1}^{y_3}\sum_{y_1=1}^{y_2} 1 = \sum_{y_{i-1}=1}^{j}\cdots\sum_{y_3=1}^{y_4} \frac{(y_3+1)!}{2(y_3-1)!}=\frac{(j+i-\rev{2})!}{(j-1)!(i-1)!}.
    \end{equation*}
    The other sums are evaluated by the same arguments.
\end{proof}

\rev{Through this result, we have as a special case the distribution of the smallest conformal p-value. We can then define a procedure that rejects $H_0^j$ when $p_j \leq \hat{t}$ for a threshold specified as the largest threshold $t$ such that $\mathbb{P}(\hat{p}_{(1)} \leq t) \leq \alpha$. By Corollary~\ref{cor:conformal_marginal_order}, 
\begin{equation*}
    \mathbb{P}(\hat{p}_{(1)} \leq t) = F_{\rm nhg}(\lfloor t(n+1)\rfloor -1; n+m, n, 1),
\end{equation*}
where $F_{\rm nhg}$ is the cumulative distribution function of a negative hypergeometric distribution. Let then $\hat{t} = \frac{F^{-1}_{\rm nhg}(\lfloor \alpha(n+1)\rfloor; n+m, n, 1) + 1}{n+1}$ where $F^{-1}_{\rm nhg}$ is the quantile function of a negative hypergeometric distribution, in which case $\mathbb{P}(\hat{p}_{(1)} \leq \hat{t}) = \lfloor \alpha(n+1)\rfloor/(n+1) \leq \alpha$. This insight is due to a discussion with Pierre Neuvial.}

Unfortunately, working with the marginals will not be sufficient in the following.

Given the explicit joint probability mass function, we present here a direct proof, which is arguably simpler than that of \cite{Angelopoulos2024:Theoretical}, that the {\rev{$\check{\rm S}$id\'{a}k}} procedure rejecting at level $1 - (1-\alpha)^{1/m}$ controls the \gls{fwer} at level $\alpha$.

\begin{theorem}
    Let assumptions~\ref{A:invariance}-\ref{A:notie} be true for a conformal score $\hats$, and consider $m$ test points $\Xbs_{n+1}, \dots, \Xbs_{n+m}$ such that $(\Xbs_{n+i})_{i\in\mathcal{H}_0}$ are jointly independent of each other. Then, the conformal \pvals{} satisfies
    \begin{equation*}
        \mathbb{P}\left(\bigcup_{j=1}^m \left\{\hat{p}_{j} \leq 1 - (1-\alpha)^{1/m}\right\}\right) \leq \alpha.
    \end{equation*}
    Additionally, as $n \to \infty$, $\mathbb{P}\left(\bigcup_{j=1}^m \left\{\hat{p}_{j} \leq 1 - (1-\alpha)^{1/m}\right\}\right) = \alpha$.
\end{theorem}
\begin{proof}
    From Lemma~\ref{lem:conformal_joint_order}, the joint probability mass function of the ordered conformal \pvals{} is
    \begin{equation*}
        \mathbb{P}(\hat{\bs{p}}_{()} = \bs{j}/(n+1)) = \mathbbm{1}[j_1\leq\cdots\leq j_m] \frac{n!m!}{(n+m)!}.
    \end{equation*}
    We have that ${\rm FWER} = \mathbb{P}\big(\bigcup_{j=1}^m \hat{p}_{j} \leq 1 - (1-\alpha)^{1/m} \big) = 1 - \mathbb{P}\big(\bigcap_{j=1}^m \hat{p}_{(j)} > 1 - (1-\alpha)^{1/m} \big)$.
    Let $t = \lfloor(1 - (1-\alpha)^{1/m})(n+1)\rfloor$ denote the \pvals{} threshold, and compute the \gls{fwer}:
    \begin{equation*}
        {\rm FWER} = 1 - \frac{n!m!}{(n+m)!}\sum_{y_m=t+1}^{n+1}\sum_{y_{m-1}=t+1}^{y_m}\cdots \sum_{y_1=t+1}^{y_2} 1 = 1 - \frac{n!}{(n+m)!} \prod_{k=0}^{m-1} (n+1-t+k)
    \end{equation*}
    which follows from the same arguments as in the proof of Corollary~\ref{cor:conformal_marginal_order}. Then, by inserting the expression for $t$
    \begin{align*}
        {\rm FWER} &= 1 - \frac{n!}{(n+m)!} \prod_{k=0}^{m-1} (n+1-\lfloor(1 - (1-\alpha)^{1/m})(n+1)\rfloor+k)\\
        &\leq 1 - \frac{n!}{(n+m)!} \prod_{k=0}^{m-1} ((1-\alpha)^{1/m}(n+1)+k)\\
        &\leq 1 - \frac{n!}{(n+m)!} (1-\alpha) \prod_{k=0}^{m-1} (n+1+k) = \alpha.
    \end{align*}
    The errors made in the inequalities goes to zero as $n \to \infty$, concluding the proof.
\end{proof}

As mentioned in the paper, if the \pvals{} are ${\rm MTP}_2$, then the Hochberg procedure controls the FWER. We will show now that the conformal \pvals{} indeed are not ${\rm MTP}_2$, but first we define the notion of ${\rm MTP}_2$.
\begin{definition}[${\rm MTP}_2$ \citep{Benjamini2001:Dependency}] \label{def:mtp2}
    A random vector $\bs{X}$ is ${\rm MTP}_2$ if for all $\bs{x}$, $\bs{y}$ we have
    \begin{equation*}
        f(\bs{x}) f(\bs{y}) \leq f({\rm min}(\bs{x}, \bs{y})) f({\rm max}(\bs{x}, \bs{y}))
    \end{equation*}
    where the minimum and maximum is component wise and $f$ is either the joint density or the joint probability distribution of $\bs{X}$.
\end{definition}
\begin{theorem}
\label{thm:not_mtp2}
    Let assumptions~\ref{A:invariance}-\ref{A:notie} be true for a conformal score $\hats$, and consider $m$ test points $\Xbs_{n+1}, \dots, \Xbs_{n+m}$ such that $(\Xbs_{n+i})_{i\in\mathcal{H}_0}$ are jointly independent of each other. Then, the conformal \pvals{} are not ${\rm MTP}_2$.
\end{theorem}
\begin{proof}
    We know the joint probability mass function, see~\eqref{eq:joint_mass}. As a counter example, consider a case where $n=2$ and $m=3$, and let $\bs{x} = [2/3, 2/3, 2/3]^\top$ and $\bs{y} = [1/3, 1, 1/3]^\top$. Then $f(\bs{x}) = 6$ and $f(\bs{y}) = 2$, while $f({\rm min}(\bs{x}, \bs{y})) = 2$ and $f({\rm max}(\bs{x}, \bs{y})) = 2$, showing that $f(\bs{x})f(\bs{y}) > f({\rm min}(\bs{x}, \bs{y})) f({\rm max}(\bs{x}, \bs{y}))$ in this case, and thereby disproving the ${\rm MTP}_2$ property.
\end{proof}

We are also interested in expressing the joint survival function of the ordered conformal \pvals{}, when a sequence of thresholds $t_1\leq t_2\leq \cdots \leq t_m$ is used. To do so, requires summing over the joint probability mass function which amounts to counting grid points on a discrete order set.

\begin{lemma}\label{lem:grid_poly}
    Given $t_1\leq\dots\leq t_m\in\{0,\dots,n\}$,
    let $P_{m-1}[y] = \sum_{i=0}^{m-1} a_i^{(m-1)} y^i$, $y \in \mathbb{N}$, be $(m-1)$-th order polynomials, defined recursively as $P_m[x] = \sum_{y=t_m+1}^{x} P_{m-1}[y]$ with initial conditions $P_{0}[y] = 1$. 
    Then, $P_{m}[n+1]$ is the number of grid points on (cardinality of) the discrete order set $\{(x_1,\dots,x_m):x_1\leq \cdots \leq x_m, x_j \in \{(t_j+1)/(n+1),\dots,1\},j=1,\dots,m\}$.
    Moreover, the polynomial coefficients are defined recursively as
    \begin{align}
        a_0^{(m)} &= -\sum_{j=0}^{m-1} \frac{a_j^{(m-1)}}{j+1} \sum_{r=0}^j \binom{j+1}{r} B_r t_m^{j+1-r},\label{eq:a0_formula}\\
        a_i^{(m)} &= \sum_{j=i-1}^{m-1} \frac{a_j^{(m-1)}}{j+1} \binom{j+1}{i} B_{j+1-i}, \quad 0 < i \leq m,\label{eq:ai_formula}
    \end{align}
    where $B_r$ is the $r$-th Bernoulli number (see \cite{Apostol2008:Bernoulli}) with the convention $B_1 = 1/2$.
\end{lemma}
\begin{proof}
    Consider the sum $P_m[n+1]=\sum_{y_m=t_m+1}^{n+1} \sum_{y_{m-1}=t_{m-1}+1}^{y_m} \cdots \sum_{y_1 = t_1 + 1}^{y_2} 1$, which is the sum over the discrete order set. Initially, we have $P_0[y_1] = 1$, followed by $P_1[y_2] = \sum_{y_1 = t_1 + 1}^{y_2} 1 = y_2 - t_1$. Faulhaber's formula (for a brief primer see \cite{Larson2019:Bernoulli}) states that
    \begin{equation*}
        \sum_{k=1}^n k^p = \frac{1}{p+1} \sum_{r=0}^p \binom{p+1}{r} B_r n^{p+1-r},
    \end{equation*}
    meaning that the sum over a polynomial of order $m-1$ is a polynomial of order $m$. By an inclusion-exclusion argument $\sum_{k=t+1}^n k^p = \sum_{k=1}^n k^p - \sum_{k=1}^t k^p$ we can handle the varying starting points of the sums. Now, by the recursive definition
    \begin{align*}
        P_m[n+1] &= \sum_{i=0}^{m-1} \Big(\sum_{y=1}^{n+1} a_i^{(m-1)} y^i - \sum_{y=1}^{t_m} a_i^{(m-1)} y^i \Big)\\
        &=\sum_{i=0}^{m-1} \frac{a_i^{(m-1)}}{i+1} \sum_{r=0}^i \binom{i+1}{r} B_r \Big((n+1)^{i+1-r} - t_m^{i+1-r}\Big).
    \end{align*}
    Collecting the terms yields the formulas in~\eqref{eq:a0_formula}-\eqref{eq:ai_formula}.
\end{proof}
This result motivates an algorithm for computing the number of grid points on the discrete order set which is only of $O(m^3)$, a drastic improvement on a naïve approach which is $O(n^m)$. In turn, this allows us to efficiently compute the joint survival function $\mathbb{P}\big(\bigcap_{j=1}^m \hat{p}_{(j)} > t_j/(n+1)\big) = \frac{n!m!}{(n+m)!} P_m[n+1]$, and then applying De Morgan's laws ${\rm FWER} = 1 - \frac{n!m!}{(n+m)!} P_m[n+1]$, meaning that for a given sequence of thresholds $t_j$, $j=1,\dots,m$, we can compute exactly the ${\rm FWER}$. Notably, for the Hochberg procedure, $t_j = \lfloor \frac{\alpha}{m - j + 1} (n+1)\rfloor$, while for the {\rev{$\check{\rm S}$id\'{a}k}} procedure it is $t_j = \lfloor (1 - (1-\alpha)^{1/m}) (n+1)\rfloor$. When $m > m_0$, the {\rev{$\check{\rm S}$id\'{a}k}} procedure can be further sharpened by first estimating $\hat{m}_0$ with a liberal estimate (for instance Storey's estimate), and then using $t_j = \lfloor (1 - (1-\alpha)^{1/\hat{m}_0}) (n+1)\rfloor$.

\section{Testing a global null hypothesis}\label{sec:global_test_appendix}
\begin{figure}[t]
    \centering
    \includegraphics[width=0.9\linewidth, page=15]{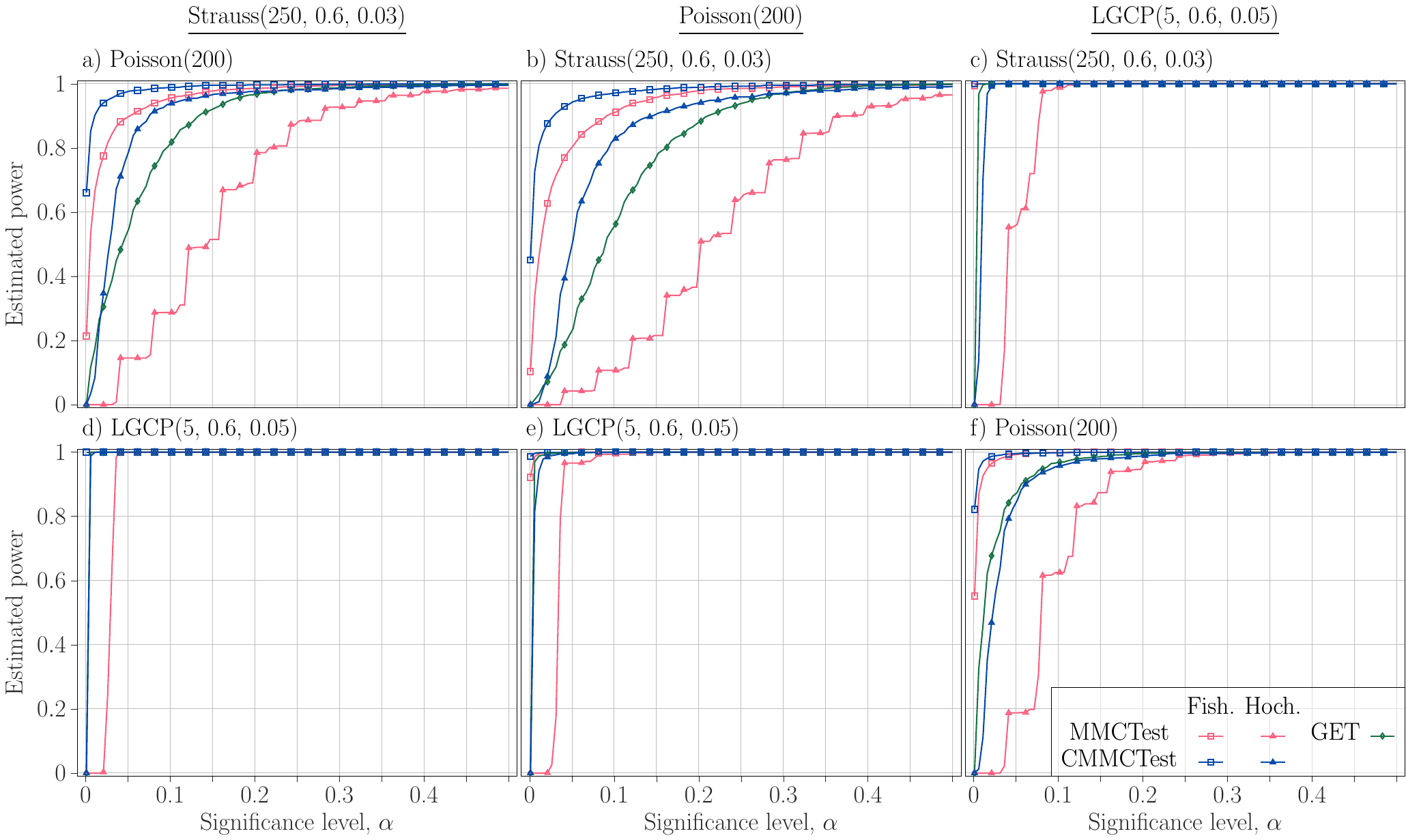}
    \caption{Estimated power curves against the significance level, $\alpha$, when testing a global null hypothesis. At the top of the figure, the null distribution is underlined, while the non-true null distribution is read above each individual plot. All the plots share the same legend. Fish.~and Hoch.~are shorthand for Fisher combination test and Hochberg procedure, respectively.}
    \label{fig:power_curves_global}
\end{figure}
\begin{figure}[t]
    \centering
    \includegraphics[width=0.95\linewidth, page=14]{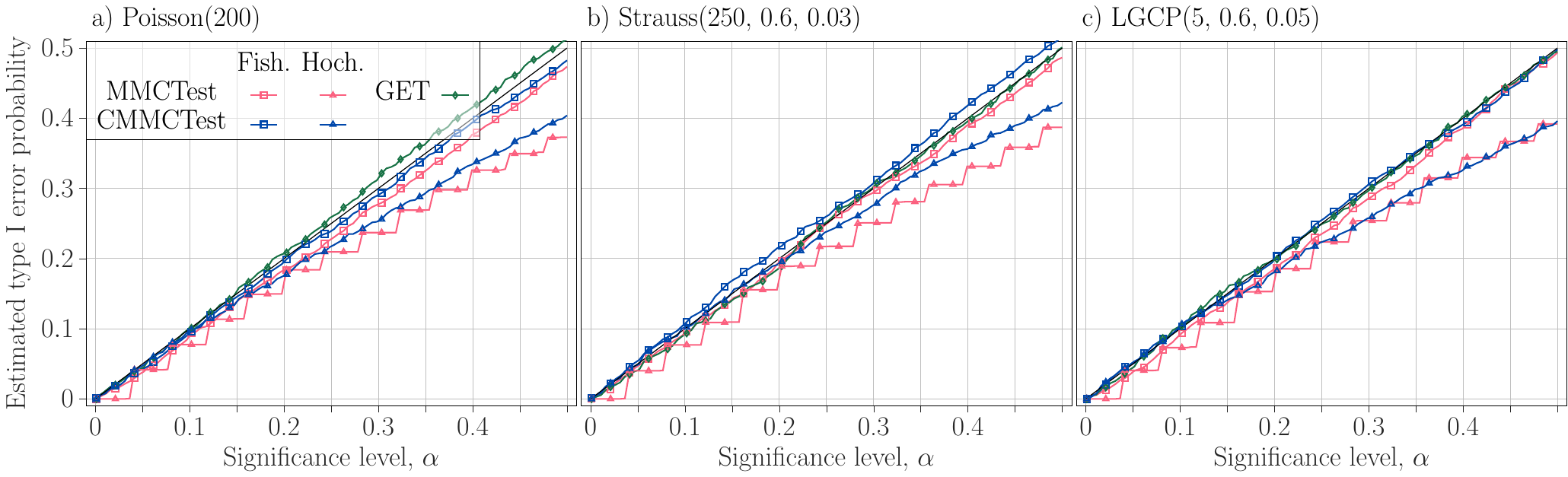}
    \caption{Estimated type I error probability curves against the significance level, $\alpha$, when testing a global null hypothesis. The black straight line shows the nominal level. All the plots share the same legend. Fish.~and Hoch.~are shorthand for Fisher combination test and Hochberg procedure, respectively.}
    \label{fig:FWER_curves_global}
\end{figure}
Here we consider testing the global null hypothesis $H_0: P_j = P_0$, for all $j=1,\dots,m$. In this case, we make no assertions regarding individual point patterns and instead test for group effects. We are interested in controlling the type I error probability at a specified significance level $\alpha \in (0, 1)$. This is the multiple testing scenario of \cite{Mrkvicka2017:Multiple} where \gls{get} is used with the concatenated functional summary statistics of the test points.
Specifically, let $T_1, \dots, T_m$ denote the functional summary statistics of the test data, then the functions are concatenated as $T = [T_1, \dots, T_m]^\top$, and a \gls{get} is made considering $T$ to be the functional summary statistic of the test data. This procedure controls the type I error probability at the nominal level. The null sample is constructed by concatenating $m$ functional summary statistics among the $n$ that were simulated at a time, resulting in a total of $n/m$ independent concatenated functional summary statistics.

In light of the methodology proposed in this work, other approaches can be taken. An immediate idea is to run the Hochberg procedure on either the \gls{mmctest} \pvals{} or the \gls{cmmctest} \pvals{} and rejecting the global null $H_0$ if one or more of the local null hypotheses are rejected. A more powerful method might be to use a combination test, for instance a Fisher combination test \cite{Fisher1925:Statistical}. The Fisher combination test works for the independent \pvals{} of the \gls{mmctest}. Theorem 2.2 of \cite{Bates2023:Testing} shows that a corrected Fisher combination test controls the type I error probability when using conformal \pvals{}, which we expect to be a very powerful method considering the power gains observed so far with the \gls{cmmctest} compared to the \gls{mmctest}.

We make a simulation study with the point processes as in Section~\ref{sec:simulation} with the scenario~\ref{enum:S1}, and show power curves in Figure~\ref{fig:power_curves_global} and type I error probability curves in Figure~\ref{fig:FWER_curves_global}. We observe significant power gains compared to the multiple GET method of \cite{Mrkvicka2017:Multiple}, particularly the most powerful procedure is the corrected Fisher combination test on the conformal \pvals{} in all the test cases.

\bibliographystyle{elsarticle-harv.bst}
\bibliography{mybib.bib}

\end{document}